\documentclass[11pt]{amsart}
\usepackage{geometry}                
\usepackage{graphicx}
\usepackage{amssymb}
\usepackage{epstopdf}
\usepackage{amsaddr}

\newcommand{\pl}{\parallel}
\renewcommand{\div}{\mbox{div\,}}
\newcommand{\curl}{\mbox{curl\,}}
\newcommand{\tr}{\mbox{tr\,}}
\newtheorem{thm}{Theorem}[section]
\newtheorem{lem}[thm]{Lemma}
\newtheorem{rem}[thm]{Remark}
\newtheorem{defn}[thm]{Definition}
\newtheorem{prop}[thm]{Proposition}
\newcommand{\R}{\mathbb{R}}
\newcommand{\Z}{\mathbb{Z}}
\renewcommand{\S}{\mathbb{S}}

\title{Some mathematics for quasi-symmetry}
\author{J.W.Burby}
\address{Los Alamos National Laboratory, Los Alamos, NM 87545, USA}
\email{jburby@lanl.gov}
\author{N.Kallinikos, R.S.MacKay}
\address{Mathematics Institute, University of Warwick, Coventry CV4 7AL, U.K.}
\email{Nikos.Kallinikos@warwick.ac.uk, R.S.MacKay@warwick.ac.uk}
\date{\today}                                           

\begin{document}

\begin{abstract}
Quasi-symmetry of a steady magnetic field means integrability of first-order guiding-centre motion.
Here we derive many restrictions on the possibilities for a quasi-symmetry.  We also derive an analogue of the Grad-Shafranov equation for the flux function in a quasi-symmetric magnetohydrostatic field.
\end{abstract}


\maketitle

\section{Introduction}

The concept of quasi-symmetry was introduced in \cite{Bo2} and then distilled into a design principle for stellarators by \cite{NZ}.
In its strongest sense it means integrability of first-order guiding-centre motion.
An excellent survey of the subject was provided by \cite{H}, assuming magnetohydrostatic (MHS) fields, that is, magnetohydrodynamic equilibrium with isotropic pressure and no mean flow.

A fundamental step was made by \cite{BQ}, who stated necessary and sufficient local conditions for integrability of guiding-centre motion in terms of a continuous symmetry of three differential forms derived from the magnetic field and made clear that quasi-symmetry can be separated from the issue of whether the magnetic field is MHS or not.

Perturbative calculations of \cite{GB}, however, make it look very likely that the only possibility for exact quasi-symmetry for MHS fields with bounded magnetic surfaces is axi\-symmetry.  Our paper gives first steps to deciding whether or not this is true.

In this paper we prove many consequences of quasi-symmetry and thereby restrictions on possible quasi-symmetric fields.  In the case of a quasi-symmetric MHS field we derive a generalisation of the axisymmetric Grad-Shafranov equation.

\cite{BQ} built in an assumption that a quasi-symmetry must be a circle-action.  Here we relax this requirement, though prove that under some mild conditions it is actually a circle-action.

We write many equations using differential forms.  For those unfamiliar with differential forms, Ch.7 of \cite{A} is a classic and there is a tutorial \cite{M} specifically for plasma physicists.

Throughout the paper we will assume enough smoothness that the equations we write make sense, at least in a weak sense.

\section{Guiding-centre motion}

We consider non-interacting charged particles in a steady, smooth (at least $C^1$) magnetic field $B$ in 3D satisfying $\div B = 0$, with $B\ne 0$ in the region of interest. 

The (non-relativistic) motion of a particle of mass $m$, charge $e$, position $q$ in a magnetic field $B$ on oriented Euclidean $\R^3$ has a formulation as a Hamiltonian system of 3 degrees of freedom (DoF), 
\begin{equation}
i_V\omega = dH
\end{equation}
for the vector field $V = (\dot{q},\dot{p})$ on the cotangent bundle $T^*\R^3$, with Hamiltonian function and symplectic form (non-degenerate closed 2-form) given by
\begin{align}
H(q,p) &= \frac{|p|^2}{2m} \label{eq:H2}\\
\omega &= -d\vartheta - e \pi^*\beta .
\end{align}
Here, $p$ is a cotangent vector at $q \in \R^3$ (applied to a tangent vector $\xi$ to $\R^3$ it produces $p(\xi) = p\cdot \xi$), $|p|$ is its Euclidean norm, $\vartheta$ is the tautological 1-form on $T^*\R^3$ defined by $\vartheta_{(q,p)}(\delta q,\delta p) = p(\delta q)$, $\pi: (q,p) \mapsto q$ is the natural map from $T^*\R^3$ to $\R^3$, $\pi^*$ is the pullback by $\pi$, and $\beta = i_B \Omega$ for volume-form $\Omega$ on $\R^3$ corresponding to the Euclidean metric and chosen orientation.  Note that $\div B=0$ is equivalent to $d\beta=0$.

One could allow time-dependent $B$, electric fields, arbitrary oriented Riemannian 3-manifold, and relativistic effects, but to focus ideas we avoid all of these (the cases with electrostatic fields and relativity are treated in an appendix).

If the perpendicular speed $v_\perp$ is less than $r_B |\Omega_B|$, where $r_B$ is the radius of curvature of the fieldlines and $\Omega_B = -e|B|/m$ is the ``gyrofrequency", then there is a locally unique ``guiding centre" $X$ within $r_B$ of $q$ and ``gyro-radius vector" $\rho$ perpendicular to $B(X)$ and smaller than $r_B$ such that
\begin{align}
v &= \frac{e}{m}B(X)\times \rho + v_\pl b(X) \\
q &= X+\rho,
\end{align}
where $v=\dot{q}$, $b = B/|B|$ and $v_\pl = v\cdot b$.  Indeed, the above formulae provide a local diffeomorphism from $(X,\rho,v_\pl)$ to $(q,v)$ for $|\rho| < r_B$. 

If $B$ varies slowly on the length-scales of $\rho$ and $v_\pl/\Omega_B$, then rotation of $\rho$ about $B(X)$ is an approximate symmetry of the particle motion.  There is a corresponding adiabatic invariant 
\begin{equation}
\mu = \frac{mv_\perp^2}{2|B(X)|} = \frac12 {|e\Omega_B(X)| |\rho|^2},
\end{equation}
called the ``magnetic moment".

If one neglects the variation of $\mu$ with time, one can reduce charged particle motion by gyro-rotation \cite{L} to obtain a Hamiltonian system of 2DoF with state $(X,v_\pl)$ and 
\begin{align}
H &= \frac12{mv_\pl^2} + \mu |B(X)|, \label{eq:H7} \\
\omega &= -e\pi^*\beta - md(v_\pl \pi^*b^\flat), \label{eq:omegaGC}
\end{align}
with $\pi^*$ now being the pullback for the map $\pi(X,v_\pl)=X$.
The equation $i_V\omega = dH$ for $V=(\dot{X},\dot{v}_\pl)$ implies
\begin{align}
e\dot{X}\times\widetilde{B} &= \mu\nabla |B| + m \dot{v}_\pl b \\
\dot{X}\cdot b &= v_\pl,
\end{align}
with the modified field
\begin{equation}
\widetilde{B} = B+ \frac{m}{e}v_\pl c, \mbox{ where } c = \curl b.
\end{equation}
These can be rearranged to give 
\begin{align}
\dot{X} &= \left({v_\pl} \widetilde{B}(X) + \frac{\mu}{e} b \times \nabla|B|\right)/\widetilde{B}_\pl \label{eq:drift}\\
\dot{v}_\pl &= -\frac{\mu}{m} \frac{\widetilde{B}}{\widetilde{B}_\pl} \cdot \nabla |B|, \label{eq:vdrift}
\end{align}
where 
\begin{equation}
\label{eq:mmB}
\widetilde{B}_\pl = \widetilde{B}\cdot b .
\end{equation}
We call (\ref{eq:drift})-(\ref{eq:vdrift}) {\em first-order guiding-centre motion} (FGCM) -- ``first-order" because, as shown in \cite{L}, it is possible to derive higher order approximations, but we will restrict attention to first-order in this paper.

The Hamiltonian formulation (\ref{eq:H7})-(\ref{eq:omegaGC}) and drift equations (\ref{eq:drift})-(\ref{eq:vdrift}) hold for an arbitrary oriented 3D Riemannian manifold, with $|\ |, \cdot, \times, \nabla, \div$ and $\curl$ interpreted appropriately.  Note that the above system is defined for $\widetilde{B}_\pl\neq0$, which is a reasonable assumption because the zeroth-order term in (\ref{eq:mmB}) is $|B|\neq0$. In toroidal geometry, however, one can treat the degeneracy at $\widetilde{B}_\pl=0$ to avoid any arising inconsistencies in gyrokinetics and derive at the same time a canonical Hamiltonian structure for the purpose of symplectic integration \cite{BE}.

The zeroth-order approximation to FGCM (using $1/e$ as convenient smallness parameter) is
\begin{align}
\dot{X} &= {v_\pl}b(X) \\
\dot{v}_\pl &= -\frac{\mu}{m} b\cdot \nabla|B| .
\end{align}
We call this ZGCM.

Both FGCM and ZGCM conserve $H$ of (\ref{eq:H7}).  We write $E$ for the value of $H$.

In ZGCM, the guiding centre moves along a fieldline.  It may be {\em circulating}, meaning $v_\pl$ has constant sign, or {\em bouncing}, meaning it is confined to an interval where $|B(X)| \le E/\mu$ and $v_\pl$ changes sign on reaching each end\footnote{the usual terminology for ``bouncing'' is ``trapped'' but this is inappropriate in a context where the whole point is to determine whether the particles are confined!}.

In FGCM, there are drifts of the guiding centre across the field.  These come from the modification $\widetilde{B}$, the $\nabla |B|$ term and the $\widetilde{B}_\pl$ denominator in (\ref{eq:drift}).
There are variants of FGCM which agree to first order in $1/e$, but we choose the one above because it has a natural Hamiltonian formulation, which we believe is important and in particular allows us to discuss its integrability.

\section{Continuous symmetries of Hamiltonian systems}
\label{sec:ctssymm}
\begin{defn}
\label{dfn:symmetry}
A {\em continuous symmetry} of a Hamiltonian system $(M,H,\omega)$ on a manifold $M$ is a $C^1$ vector field $U$ on $M$ such that the Lie derivatives $L_UH$ and $L_U\omega$ are both zero.
\end{defn}
It follows that there is a conserved quantity locally, and globally under mild conditions.  This is a Hamiltonian version of Noether's theorem.

\begin{thm}
\label{thm:Noether}
If $U$ is a continuous symmetry for a Hamiltonian system $(M,H,\omega)$ with vector field $V$ then there is a conserved local function $K$ for $V$. 
If there are combinations $f U + g V$ of  $U$ and $V$ with closed or recurrent trajectories realising a basis of first homology $H_1(M)$ then $K$ is global. 
\end{thm}

\begin{proof}
$L_U\omega=0$ and $d\omega=0$ imply $d i_U \omega=0$, so by Poincar\'e's lemma $i_U \omega = dK$ for some local function $K$, and then 
\begin{equation}
i_V dK = i_V i_U \omega = -i_U dH = -L_UH= 0.
\label{eq:iVdK}
\end{equation}

If there is a combination $w= f U + g V$ of $U$ and $V$ with a closed trajectory $\gamma$ then  
\begin{equation}
\int_\gamma i_U \omega = \int_0^T (f i_U + g i_V) i_U\omega \ dt,
\end{equation}
where $t$ is time along $w$ and $T$ is the period.
The first term vanishes by antisymmetry of $\omega$ and the second because of (\ref{eq:iVdK}).  For a recurrent trajectory, close it by a short arc and bound the error to obtain that the integral of $i_U\omega$ in its homology direction is zero (the concept of homology direction is described in \cite{Fr}).  If $\int_\gamma i_U\omega = 0$ holds for $\gamma$ representing a basis of $H_1(M)$ we deduce that $K$ is global.
\end{proof}

\begin{defn}
\label{def:integrable}
A 2DoF Hamiltonian system with vector field $V$ is {\em integrable} if it has a continuous symmetry $U$ with global conserved quantity $K$ and $U,V$ are linearly independent almost everywhere (a.e.) (equivalently $dK, dH$ are linearly independent a.e.).
\end{defn}

Note that Definition \ref{dfn:symmetry} implies that the symmetry $U$ and the Hamiltonian vector field $V$ commute, because $i_{[U,V]}\omega = L_Ui_V\omega-i_VL_U\omega = L_UdH = dL_UH = 0$ and $\omega$ is non-degenerate.

For an integrable 2DoF system, the bounded regular\footnote{A component $C$ of a level set of a $C^1$ function $F:M\to N$ is {\em regular} if $DF$ is surjective everywhere on $C$; in the present context, $F=(K,H)$ and $N=\R^2$.} components of level sets of $(K,H)$ are 2-tori and there is a coordinate system in which $U, V$ are both constant vector fields on each of them.  This is a special case of the Arnol'd-Liouville theorem \cite{A}.  Here is a statement and proof. 

\begin{thm}
\label{thm:AL}
If $U,V$ are commuting vector fields on a bounded surface $S$, independent everywhere on it, then $S$ is a 2-torus and there are coordinates $(\theta^1, \theta^2)\!\!\mod 2\pi$ on it in which $U,V$ are constant.
\end{thm}

\begin{proof}
Let $\phi^U$ be the flow of $U$ and $\phi^V$ the flow of $V$.  For $t=(t_1,t_2) \in \R^2$ let $\phi_t = \phi_{t_1}^U \circ \phi_{t_2}^V$. Because the two commute and are independent, $\phi$ is a transitive action of the group $\R^2$ on $S$.  Choose a point $0 \in S$ and let $T$ be the set of $t \in \R^2$ such that $\phi_t(0)=0$.  It is a discrete subgroup of $\R^2$ (as a group under addition).  Then $S$ is diffeomorphic to $\R^2/T$.  Since $S$ is bounded, $T$ must be isomorphic to $\Z^2$.  Thus $T$ is generated by a pair $(T^1, T^2)$ of independent vectors in $\R^2$.  Let $A$ be the matrix with columns $(T^1,T^2)$.  Then we obtain an action of $\S^1\times\S^1$ (with $\S^1 = \R/2\pi \Z$) on $S$ by $\theta \mapsto \phi_{A\theta/2\pi}(0)$ where $\theta = (\theta^1,\theta^2)\in \S^1\times \S^1$.  Keeping $0\in S$ fixed this action defines a diffeomorphism $\mathbb{S}^1\times\mathbb{S}^1\rightarrow S$.  In these coordinates, $U$ is the first column of $2\pi A^{-1}$ and $V$ is its second column, thus constant vector fields.
\end{proof}

\begin{defn}
Coordinates $(\theta^1,\theta^2)\!\!\mod 2\pi$ on a 2-torus in which commuting vector fields $U,V$ on it are constant are called {\em Arnol'd-Liouville (AL) coordinates}.
\end{defn}

The concepts of integrability and AL coordinates have higher dimensional analogues but the 2DoF context suffices here.

\section{Quasi-symmetry}

\begin{defn}
\label{def:qs}
Given a magnetic field $B$ on an oriented 3D Riemannian manifold $Q$, a vector field $u$ on $Q$ is a {\em quasi-symmetry} of $B$ if $U=(u,0)$ is a continuous symmetry for FGCM for all values of magnetic moment $\mu$.  
\end{defn}

We assume $B$ nowhere zero on $Q$ in order for FGCM to make sense.
Note that in contrast to most of the literature (e.g.~\cite{H}) we do not assume that $B$ is MHS.  Indeed, one might like to apply the concept of quasi-symmetry to magnetohydrodynamic equilibria with a mean flow (cf.~\cite{SH}) or with anisotropic pressure, for example.  The concept of FGCM does not require an MHS field, so neither should quasi-symmetry.

A simple example of a quasi-symmetric magnetic field is any axisymmetric $B$ in Euclidean space.  Take $u=\partial_\phi = r \hat{\phi}$ in cylindrical coordinates $(r,\phi,z)$.  Axisymmetry of $B$ can be defined in various ways, e.g.~$L_uB=0$ or $L_u\beta = 0$ for this $u$.  They are equivalent because $\div u = 0$ and 
\begin{equation}
i_{[u,B]}\Omega = L_u\beta - (\div u)\beta.
\label{eq:uB}
\end{equation}

Our first main theorem is:

\begin{thm}
\label{thm:main}
A vector field $u$ is a quasisymmetry of a magnetic field $B$ iff
\begin{align}
L_u|B| &= 0 \\
L_u\beta &= 0 \\
L_u b^\flat &= 0.
\end{align}
\end{thm}

\begin{proof}
Recall the Hamiltonian and symplectic form for FGCM: 
\begin{align}
H &= \frac12 m{v_\pl^2} + \mu |B(X)| \\
\omega &= -e\pi^*\beta - d(p_\pl \pi^*b^\flat) .
\end{align}
Then $L_UH = \mu L_u|B|$, so $L_UH=0$ for all $\mu$ iff $L_u|B|=0$.
Next 
\begin{equation}
L_U\omega = -e L_u\beta - md(v_\pl L_u b^\flat) = -e L_u\beta - mdv_\pl \wedge L_ub^\flat - mv_\pl dL_ub^\flat.
\end{equation}
Apply this to an arbitrary pair of tangents to $Q$ and set $v_\pl=0$ to deduce that $L_U\omega=0$ implies $L_u\beta=0$.
Apply it to an arbitrary tangent $\xi$ to $Q$ and the vector $(0,1)$ tangent to $Q\times \R$ to deduce that $i_\xi L_ub^\flat = 0$, so $L_ub^\flat=0$.

In the other direction, if $L_u\beta=0$ and $L_ub^\flat=0$ then $L_U\omega=0$.

So we have proved that $u$ is a quasi-symmetry of $B$ iff $L_u|B|=0$, $L_u\beta=0$ and $L_ub^\flat=0$.
\end{proof}

We write the three conditions of Theorem~\ref{thm:main} in vector calculus or suffix notation for comparison:
\begin{align}
u\cdot \nabla |B| &= 0 \\
\curl(B\times u) &= 0 \\
u^i \partial_ib_j + b_i \partial_j u^i &= 0 . \label{eq:Lubflat}
\end{align}

Quasi-symmetry has the following significant consequences. 

\begin{thm}
\label{thm:conseq}
If $u$ is a quasi-symmetry of a magnetic field $B$ then $L_uB^\flat=0$, 
$L_u\Omega=0$ and $L_uB = 0$.
\end{thm}

\begin{proof}
To prove $L_uB^\flat=0$, use $B^\flat = |B| b^\flat$, so 
\begin{equation}
L_uB^\flat = (L_u|B|) b^\flat + |B| L_u b^\flat.  
\label{eq:LuBflat}
\end{equation}
By Theorem~\ref{thm:main}, $L_u|B|=0$ and $L_ub^\flat=0$.  So $L_uB^\flat=0$.

To prove $L_u\Omega=0$ note that $\beta \wedge b^\flat = |B| \Omega$.  Applying $L_u$, we obtain 
\begin{equation}
L_u\beta \wedge b^\flat + \beta \wedge L_u b^\flat = (L_u |B|) \Omega + |B| L_u\Omega .
\label{eq:LubetawedgeBflat}
\end{equation}
According to Theorem~\ref{thm:main}, the first three terms of this are zero.  As $|B| \ne 0$, we obtain $L_u\Omega=0$.

To prove that $L_uB=0$, note that it can alternatively be written as $[u,B]=0$.  Use the formula
\begin{equation}
i_{[u,B]}\Omega = L_u i_B \Omega - i_B L_u\Omega ,
\label{eq:iuBOmega}
\end{equation}
which holds for any pair of vector fields $u, B$ and any differential form $\Omega$.
By Theorem~\ref{thm:main}, $L_u\beta=0$, and we just proved that $L_u\Omega=0$.  So using $\Omega$ non-degenerate, we see that $[u,B]=0$, cf.~(\ref{eq:uB}).
\end{proof}

In suffix notation, the first result of Theorem~\ref{thm:conseq} is written analogously to that for $L_u b^\flat$ in (\ref{eq:Lubflat}).
Alternatively and more usefully it can be written in vector calculus as 
\begin{equation}
\label{eq:LubigBflat}
u\times J = \nabla (u\cdot B),
\end{equation}
because
\begin{equation}
L_uB^\flat = i_udB^\flat + di_uB^\flat = i_ui_J\Omega + d(u\cdot B).
\label{eq:iuiJ1}
\end{equation}
The second says $\div u = 0$, and for the third, 
\begin{equation}
L_uB = [u,B]=u\cdot\nabla B - B\cdot\nabla u = \curl(B\times u) + (\div B) u - (\div u) B.  
\end{equation}
Since $\div B =0$ and we already proved that $\div u=0$, then $[u,B]=0$ can be written in this case as $\curl(B\times u)=0$.

Noting that some steps in the above proof are reversible, we can derive various alternative necessary and sufficient conditions for quasi-symmetry.  The following theorem gives some examples, from which we shall frequently use (i) or (ii).  Case (i) is a slight generalisation of the formulation in \cite{BQ}.

\begin{thm}
\label{thm:alt}
A vector field $u$ is a quasi-symmetry of a magnetic field $B$ iff any of the following sets of conditions hold:
\begin{enumerate}
\item $L_u |B| = 0$, $L_u\beta = 0$, $L_u B^\flat = 0$;
\item $L_u\Omega=0$, $L_u\beta = 0$, $L_u B^\flat=0$;
\item $L_u\Omega=0$, $L_uB=0$, $L_uB^\flat=0$.
\end{enumerate}
\end{thm}

\begin{proof}
(i) To prove the first set, we use (\ref{eq:LuBflat}).  Thus under $L_u|B|=0$ and $B \ne 0$, we obtain $L_uB^\flat=0$ iff $L_ub^\flat=0$, which converts Theorem~\ref{thm:main} to (i).

(ii) The second comes from the first and (\ref{eq:LubetawedgeBflat}).

(iii) The third comes from the second and (\ref{eq:iuBOmega}).
\end{proof}

Here are some additional consequences of quasi-symmetry.

\begin{thm}
\label{thm:add}
If $u$ is a quasi-symmetry of $B$ then
\begin{enumerate}
\item $L_u(u\cdot B)=0, L_u(u\cdot b)=0,$
\item $[u,b]=0, [u,B/|B|^2]=0, [u,u_\perp]=0$ (where $u_\perp$ is the component of $u$ perpendicular to $B$),
\item $[u,J]=0$ (where $J=\curl B$), $[u,[J,B]]=0$, $L_u(J\cdot B)=0$, $L_J(u\cdot B)=0$.
\end{enumerate}
\end{thm}

\begin{proof}
(i) $u\cdot B = i_u B^\flat$ so $L_u(u\cdot B) = i_u L_uB^\flat +i_{[u,u]}B^\flat$, both of which are zero.

For $u\cdot b$, apply $L_u$ to $u.B = |B| u\cdot b$ and use the above plus $B \ne 0$ to deduce that $L_u(u\cdot b)=0$.

(ii) For $[u,b]$, use $[u,B]=0$, $B=|B| b$ and $L_u|B|=0$, to obtain $[u,b]=0$.

Similarly, $[u,B/|B|^2]=0$ or indeed $[u, f(|B|)B]=0$ for any function $f$.

$u_\perp = u - (u\cdot B) B/|B|^2$ and $L_u$ on each of these terms is zero, so $L_u u_\perp = 0$.

(iii) $J=\curl B$ translates to $i_J\Omega = dB^\flat$.  Apply $L_u$ to each side.  $L_u i_J\Omega = i_J L_u \Omega + i_{[u,J]}\Omega$ and $L_u dB^\flat = dL_uB^\flat=0$.  But $L_u\Omega=0$ so we deduce that $i_{[u,J]}\Omega=0$.  $\Omega$ is non-degenerate, so $[u,J]=0$.

For $[u,[J,B]]$, we use the Jacobi identity $[u,[J,B]] + [J,[B,u]] + [B,[u,J]] = 0$.  We already proved that $[B,u]=0$ and $[u,J]=0$.  So $[u,[J,B]]=0$.

$L_u(J\cdot B) = L_u i_J B^\flat = i_J L_u B^\flat = 0$, using $[u,J]=0$.

$L_J(u\cdot B) = i_Jd(u\cdot B) = i_JL_uB^\flat = 0$, using (\ref{eq:iuiJ1}).

\end{proof}

\section{Flux function}

The condition $L_u\beta=0$ of Theorem~\ref{thm:main} merits additional comment.  We discuss it in a more general context than quasi-symmetry.  Specifically, we require only $L_u\beta=0$, $\div B=0$ and $\div u = 0$.

Because $d\beta=0$ and $\beta = i_B\Omega$, $L_u\beta=0$ is equivalent to $di_ui_B\Omega=0$.  Thus by Poincar\'e's lemma $i_ui_B\Omega = d\psi$ for some function $\psi$ locally (in vector calculus, $B\times u = \nabla \psi$), and both $u$ and $B$ are tangent to regular level sets of $\psi$.  

An important question is whether $\psi$ is global.  It is global if there are combinations $fu+gB$ with closed or recurrent trajectories realising a basis of $H_1(Q)$.  For the case of $Q$ being a solid torus with a circulating magnetic field, then $B$ has a closed trajectory realising $H_1(Q)$ so $\psi$ is global.  For more complicated domains, it might fail.

\begin{defn}
A {\em flux function} for a field $B$ on $Q$ is a globally defined function $\psi: Q \to \R$ with $i_Bd\psi = 0$ and $d\psi \ne 0$ a.e.
\end{defn}
Note that existence of a flux function $\psi$ is an assumption of the standard approach to quasi-symmetry (e.g.~\cite{H}), whereas here we derived it as a consequence, at least as a local function.
For many purposes, however, we will need to assume that $\psi$ is global and has non-zero derivative a.e. (see ahead to Definition~\ref{def:Bqs}).
Note that in \cite{H}, $\psi$ is chosen to be the toroidal flux enclosed by the level set and ``flux function" is used for any function of $\psi$.

If $\psi$ is global, it follows from the classification of surfaces that bounded regular components of level sets of $\psi$ are 2-tori.    
\begin{defn}
The bounded regular components of level sets of a flux function are called {\em flux surfaces}.
\end{defn}
Furthermore, $u$ and $B$ are independent everywhere on such a 2-torus, because $i_u i_B \Omega = d\psi$.  $L_u\beta=0$ with $\div u = 0$ imply that $[u,B]=0$, cf.~(\ref{eq:uB}).  Using Theorem~\ref{thm:AL}, it follows that there are coordinates $(\theta^1,\theta^2)\mod 2\pi$ on the 2-torus in which both $u$ and $B$ are constant vector fields.  
There are many works on ``flux coordinates", including the book \cite{DHCS} and the recent paper \cite{KG}, but we are not aware of any of them using this very natural AL approach.  The closest we have seen is \cite{Ha}.

We now derive a formula for the winding ratios of $X=u, B$ on a flux surface.
\begin{defn} 
The {\em winding ratio} $\iota_X$ of a vector field $X$ on a 2-torus with coordinates $(\theta^1,\theta^2) \in \S^1\times \S^1$ is
the limit as $t\to \infty$ of the ratio of the number of revolutions made in $\theta^1$ along a trajectory of vector field $X$, to that in $\theta^2$.  $\iota_X$ is considered as a point in the projective line $\R P^1$, to include the option of $\infty$ and to ignore the sign of $X$.
\end{defn}
For vector fields $X$ of ``Poincar\'e type'' (those having a cross-section) on a 2-torus, the limit exists and is the same for all trajectories and for both signs of $X$.
Since $u$ and $B$ are conjugate to non-zero constant vector fields on each flux surface, they are of Poincar\'e type.

As $X$ conserves $\Omega$ and $\psi$, it also conserves an area-form on flux surfaces.  Indeed, let
\begin{equation}
\mathcal{A} = i_n\Omega, \mbox{ where } n = \frac{\nabla \psi}{|\nabla \psi|^2}.
\end{equation}
In vector calculus, $\mathcal{A}(\xi,\eta) = n\cdot (\xi \times \eta)$.  Then the restriction $\mathcal{A}_C$ of $\mathcal{A}$ to
a regular component $C$ of a level set of $\psi$ is non-degenerate and conserved by $X$.  To prove the conservation it is enough to work out $L_X\mathcal{A}$ on $(u,B)$, which form a basis of tangents to $C$.  Using $[u,B]=0$, we have
\begin{equation}
i_ui_BL_X\mathcal{A} = L_X i_u i_B i_n\Omega = L_X i_n d\psi = L_X 1 = 0.
\end{equation}

\begin{thm}
\label{thm:wr}
If $\div u = \div B = 0$ and $i_u i_B \Omega = d\psi$ then for $X=u$ or $B$ on a bounded regular component $C$ of a level set of $\psi$,
\begin{equation}
\iota_X = -\frac{\int_{\gamma_1} i_X\mathcal{A}_C}{\int_{\gamma_2}i_X\mathcal{A}_C},
\label{eq:iota}
\end{equation}
where $\gamma_j$ is any closed loop on $C$ making one turn in $\theta^j$ and none in the other.
\end{thm}

\begin{proof}
As $L_X\mathcal{A}_C = 0$, we deduce that $i_X\mathcal{A}_C$ is closed so its integral round a closed loop $\gamma$ depends on only the homology class $[\gamma]$ of the loop.
Take a long piece of trajectory of $X$ on $C$ and close it by a short arc on $C$, making a closed loop $\gamma$.  It has homology class close to $N([\gamma_1] + \iota_X[\gamma_2])$ for some large integer $N$.  $i_X\mathcal{A}_C(\dot{\gamma})$ is zero except on the short arc.  Taking the limit we obtain
\begin{equation}
\int_{[\gamma_1]+\iota_X [\gamma_2]} i_X\mathcal{A}_C = 0.
\end{equation}
Hence the formula of the theorem.
\end{proof}

The same formula applies to the current density $J$ for an MHS field with $p$ constant on flux surfaces.

\section{The invariant tori of FGCM}
Let us compute the conserved quantity $K$ of FGCM resulting from quasi-symmetry.  Recall from Theorem~\ref{thm:Noether} that $K$ results from $i_U\omega=dK$.  Recall from (\ref{eq:omegaGC}) that $\omega = -e\pi^*\beta - m d(v_\pl \pi^*b^\flat)$ and from Definition~\ref{def:qs} that $U=(u,0)$.  So
\begin{equation}
i_U\omega = -ei_u\beta -m L_U(v_\pl \pi^*b^\flat) + m di_u(v_\pl b^\flat) = dK,
\end{equation}
with
\begin{equation}
K = -e\psi - m v_\pl u\cdot b,
\end{equation}
using $L_U(v_\pl \pi^* b^\flat)=v_\pl L_u b^\flat = 0$.
In particular, we see that $K$ is global iff $\psi$ is global.  

$K$ governs how far particles move from a flux surface.
Using conservation of $K$, we see that 
\begin{equation}
v_\pl = -\frac{e\psi+K}{m u\cdot b}.
\label{eq:vpl}
\end{equation}
Hence by conservation of $H$ in (\ref{eq:H7}), the tori for FGCM are given in projection to guiding-centre position by
\begin{equation}
\frac12 \frac{(e\psi+K)^2}{m(u\cdot b)^2} + \mu |B| = E,
\end{equation}
with parallel velocity recovered by (\ref{eq:vpl}). 
This can be written as
\begin{equation}
\psi = -\frac{K}{e} \pm \frac{u\cdot b}{e}\sqrt{2m(E-\mu|B|)}.
\end{equation}
We see the same division of motion into circulating and bouncing, as for ZGCM.  Note that by $L_u|B|=0$, the set of $X$ where $|B(X)|>E/\mu$ is a set of $u$-lines.

\section{Effect on the metric}

Next we examine the relation of a quasi-symmetry $u$ to the Riemannian metric $g$.  The conjecture of \cite{GB} is equivalent to $u$ being a Killing field for Euclidean metric $g$.  
\begin{defn}
A vector field $u$ is a {\em Killing field} for a Riemannian metric $g$ if $L_ug=0$.
\end{defn}
In Euclidean space $L_ug=0$ can be written as $\nabla u+(\nabla u)^T=0$.

We have not managed to prove or disprove that a quasi-symmetry is a Killing field yet, but the following theorem gets two-thirds of the way (by showing the subspace of possibilities for $L_ug$ at a point is constrained to a codimension-4 subspace of the 6D space of symmetric $3\times 3$ matrices).   It applies to an arbitrary Riemannian metric $g$.

\begin{thm}
\label{thm:Lug}
Let $u$ be a quasi-symmetry for magnetic field $B$, with flux function $\psi$.
Where $u,B$ are independent, then $d\psi \ne 0$ and $(B,u,n)$ is a basis, with $n=\nabla \psi / |\nabla \psi|^2$. 
With respect to this basis, $L_u g$ has matrix
\begin{equation}
\left[\begin{array}{ccc}
0 & 0 & 0 \\
0 & L_u|u|^2 & i_n L_u u^\flat \\
0 & u\cdot [n,u] & L_u |\nabla \psi|^{-2} 
\end{array} \right]
\label{eq:Lug}
\end{equation}
and the diagonal terms are related by
\begin{equation}
L_u |\nabla \psi|^{-2} = -\frac{|B|^2}{|\nabla \psi|^4} L_u |u|^2.
\label{eq:diag}
\end{equation}
\end{thm}

\begin{proof}
$d\psi \ne 0$ where $u,B$ are independent because $i_u i_B \Omega = d\psi$ and $\Omega$ is nondegenerate.
$\Omega(B,u,n) = i_n i_u i_B \Omega = i_n d\psi = 1 \ne 0$, so $(B,u,n)$ is a basis.

The calculation of $L_ug$ makes use of Lemma~\ref{lem:1} to follow, which says that the standard commutation relation $i_XL_u = L_ui_X - i_{[u,X]}$ applies not only to differential forms but also to any covariant 2-tensor, thus including the case of a metric tensor.
We apply this to the metric tensor $g$ for $X=B,u,n$ in turn.

For $X=B$ we obtain $i_BL_ug = 0$, so the first row of $L_ug$ is zero, and also the first column by symmetry of $g$.

For $X=u$ we obtain $i_u L_u g = L_u i_u g = L_u u^\flat$.  The diagonal component is obtained by contracting this with $u$: $i_u L_u u^\flat = L_u i_u u^\flat = L_u |u|^2$.  

For $X=n$ we obtain 
\begin{equation}
i_n L_u g = L_u i_n g - i_{[u,n]}g.  
\label{eq:inLug}
\end{equation}
For the off-diagonal term we contract (\ref{eq:inLug}) with $u$.  Firstly, $i_u L_u i_n g = L_u i_ui_n g = L_u (u\cdot n)$ but $u\cdot n=0$ from $i_u d\psi = 0$.  Secondly, $i_u i_{[u,n]}g = u\cdot [u,n]$.  
For the diagonal term, we contract (\ref{eq:inLug}) with $n$.  Firstly,
\begin{equation}
i_n L_u i_n g = i_nL_un^\flat = i_n L_u (d\psi/|\nabla \psi|^2) = |\nabla \psi|^{-2} i_nL_u d\psi + (L_u |\nabla\psi|^{-2}) i_n d\psi = L_u |\nabla\psi|^{-2}, 
\end{equation}
using $L_u\psi=0$ and $i_n d\psi=1$.  Secondly,
\begin{equation}
i_ni_{[u,n]}g = i_{[u,n]}n^\flat = L_ui_n n^\flat - i_nL_un^\flat =  L_u|\nabla\psi|^{-2} - |\nabla \psi|^{-2}i_nL_ud\psi - (L_u |\nabla\psi|^{-2}) i_n d\psi = 0.
\label{eq:iniung}
\end{equation}

Finally, we prove the indicated relation between the diagonal terms.
Using $B\times u = \nabla \psi$,
\begin{equation}
|\nabla \psi|^2 = |B|^2 |u|^2 - (u\cdot B)^2.
\end{equation}
We proved $L_u(u\cdot B) = 0$ and $L_u|B|=0$, so applying $L_u$ to the above equation, we obtain $L_u |\nabla \psi|^2 = |B|^2 L_u|u|^2$.  Hence the result.
\end{proof}

Because $g$ is symmetric, $L_ug$ is symmetric but we give the two alternative expressions for the off-diagonal components in (\ref{eq:Lug}) and one can check they are equal ($u\cdot [n,u] = i_{[n,u]}u^\flat = i_nL_uu^\flat-L_ui_nu^\flat$ but $n\cdot u=0$).  

One could choose other normalisations of $\nabla\psi$, but an advantage of the chosen one is that $n\cdot [u,n]=0$, as proved in (\ref{eq:iniung}).

The relation (\ref{eq:diag}) can alternatively be obtained by using $\div u=0$ and the local expression $\Omega = \pm \sqrt{\det g}\,dx^1\wedge dx^2 \wedge dx^3$ with $\det g$ being the determinant of the matrix representing $g$ in coordinate system $(x^1,x^2,x^3)$ (i.e.~$g(\xi,\eta) = g_{ij} \xi^i \eta^j$).  A calculation shows that $\div u = \frac12 \tr(g^{-1}L_ug)$ and hence (\ref{eq:diag}).

The previous theorem highlights the importance of $L_u u^\flat$, not just for the off-diagonal term but also because $L_u |u|^2 = i_u L_u u^\flat$.  So

\begin{thm}
\label{thm:killingqs}
A quasi-symmetry $u$ is a Killing field iff $L_u u^\flat=0$.
\end{thm}

In vector calculus $L_uu^\flat=0$ can be written as $w=0$, where $w=v\times u+\nabla|u|^2$ with $v=\curl u$. The relation between the two is $w^\flat=L_uu^\flat$.

We conclude this section with the required lemma.

\begin{lem}
\label{lem:1}
For any covariant 2-tensor $g$, and vector fields $u,X$,
\begin{equation}
i_{[u,X]} g = L_u i_X g - i_X L_u g.
\end{equation}
\end{lem}

\begin{proof}
For arbitrary vector fields $u$, $X$ and $Y$, and covariant 2-tensor $g$,
\begin{equation}
(L_ug)(X,Y) = L_u(g(X,Y)) - g(L_uX,Y) - g(X, L_uY).
\end{equation}
This says
\begin{equation}
i_Yi_XL_ug = L_ui_Yi_Xg - i_Yi_{[u,X]}g - i_{[u,Y]}i_Xg.
\end{equation}
Now $i_Xg$ is a differential form, so the usual commutation relation
\begin{equation}
L_ui_Y i_Xg = i_Y L_u i_Xg + i_{[u,Y]}i_X g
\end{equation}
can be employed for the first term on the right.  
It results that
\begin{equation}
i_Yi_XL_ug = i_Y L_u i_Xg - i_Y i_{[u,X]}g.
\end{equation}
This is true for all $Y$, hence the result.
\end{proof}

\section{Circle action}

In the case of axisymmetry, the trajectories of $u$ are all closed and have a common period (single points on the axis of symmetry and circles elsewhere, of period $2\pi$).  We say the flow of $u$ generates a circle action.  

\begin{defn}
A {\em circle action} on a manifold $M$ is a differentiable mapping $\Phi:\S^1 \times M \to M$, $(\theta,x) \mapsto \Phi_\theta(x)$ such that $\Phi_0(x) = x$ and $\Phi_{\theta+\theta'}(x) = \Phi_\theta(\Phi_{\theta'}(x))$ for all $\theta,\theta' \in \S^1$ and $x \in M$. 
\end{defn}
The orbit of a point $x$ under a circle action is either a point or diffeomorphic to a circle.

\cite{BQ} formulated quasi-symmetry in terms of a circle action preserving FGCM.  Given a circle action one can obtain a vector field $u = \partial_\theta \Phi_\theta |_{\theta=0}$.  It is a quasi-symmetry if $\Phi$ preserves FGCM.  In our treatment of quasi-symmetry here, we do not require the trajectories of a quasi-symmetry to be circles, but we prove now that under mild conditions any quasi-symmetry does generate a circle action locally.


\begin{thm}
If $u$ is a quasi-symmetry for $B$, $\psi$ is global, $S$ is a bounded regular component of a level set of $\psi$ (flux surface), and S contains a regular component $T$ of a joint level set of either
\begin{enumerate}
\item $(|B|,\psi)$ with $u\cdot B \ne 0$, or
\item $(u\cdot B,\psi)$,
\end{enumerate}
then $T$ is a circle and a closed $u$-line, and all $u$-lines on $S$ are closed and of the same period.  Furthermore, all nearby flux surfaces have the same properties and the same period.  The circles are non-contractible on the flux surfaces and all have the same rational winding ratio on this interval of flux surfaces.
\end{thm}

\begin{proof}
The vector field $u$ is nowhere zero on $S$ because $i_ui_B\Omega = d\psi \ne 0$.
A bounded regular component $T$ of a level set of two functions in 3D is a circle.  $L_u \psi = 0$.  In case (i), $L_u|B|= 0$ and $d|B|, d\psi$ are independent, so $T$ is a $u$-line.  In case (ii), $L_u(u\cdot B)=0$ and $d(u\cdot B), d\psi$ are independent, so $T$ is a $u$-line.  

Now $[u,B]=0$ so by Theorem~\ref{thm:AL}, $u$ is conjugate to a constant vector field on $S$. Because one $u$-line on $S$ is closed, it follows that all $u$-lines on $S$ are closed and have the same period.  

Independence of $d\psi$ and $d|B|$ (respectively $d(u\cdot B)$) on $T$ implies the same for all nearby components of level sets of $(\psi,|B|)$ (respectively $(\psi,u\cdot B)$).  So we obtain the same result for all nearby flux surfaces.

To prove the period of the $u$-lines is the same for nearby flux surfaces, we treat the two cases separately.

In case (i), let the function $f(x)=2\pi/\tau(x)$ where $\tau(x)$ is the period of the $u$-line through the given point $x$, and define vector field $R = u/f$.  Then $R$ generates a circle action $\Phi$.  Define the circle-average $\langle \omega \rangle$ of any differential form or vector field $\omega$ by
\begin{equation}
\langle \omega \rangle = \frac{1}{2\pi} \int_0^{2\pi} \Phi^*_\theta \omega \ d\theta.
\label{eq:circleavg}
\end{equation}
Note that the circle-average of $L_R \omega$ is zero for any $\omega$.  From $u = f R$ follows
\begin{equation}
0=L_uB^\flat = (B\cdot R)\ df + f L_R B^\flat .
\end{equation}
Take the circle-average of this equation to obtain
\begin{equation}
0=(B\cdot R)\ df,
\end{equation}
because $f$ and $B\cdot R$ are constant along $\Phi$-orbits.
So if $u\cdot B \ne 0$ then $df = 0$.  This applies on a neighbourhood of $S$, so $f$ is locally constant.

In case (ii), $L_uB^\flat=0$ implies that $i_u i_J\Omega = -d(u\cdot B)$ (equation (\ref{eq:iuiJ1})) and $[u,J]=0$ (Theorem~\ref{thm:add}(iii)).  So $u,J$ are commuting vector fields on each level set of $u\cdot B$.  Now $d(u\cdot B) \ne 0$, so the $u$-lines on the level set of $u\cdot B$ are closed and have the same period (using Theorem~\ref{thm:AL}).  By independence of $d(u\cdot B)$ and $d\psi$ this gives us a $u$-line on each nearby flux surface and it has the same period.  Thus they all have the same period.

The $u$-lines are non-contractible on the flux surfaces because of the conjugacy to a constant vector field.  The winding ratio of $u$ as a function of $\psi$ is continuous and rational so it is constant.
\end{proof}

Note that case (ii) can not occur for an MHS field with $p$ constant on flux surfaces because of Theorem~\ref{thm:u.B} to follow, but it might be useful in other situations.


The rational winding ratio $m:n$ of the $u$-lines is called the {\em type} of the quasi-symmetry.  With $\theta^1$ poloidal and $\theta^2$ toroidal, $m:n=0:1$ is called quasi-axisymmetric (QA), $1:0$ is called quasi-poloidal (QP), and anything else is called quasi-helical (QH).

Note that both quasi-symmetry and circle action allow the possibility of short fibres, for example, a region of $m:n$ QH may shrink onto a closed $u$-line (which will be also a $B$-line) around which the rest have winding ratio $m:n$ (a ``Seifert fibration").  If the $u$-period of the main $u$-lines is $2\pi$, then the period of the short fibre is $2\pi/n$.

Note also that the construction in case (ii) of tori with $u\cdot B$ constant supporting commuting vector fields $u,J$ applies for a general quasisymmetry $u$, without requiring the derivatives of $u\cdot B$ and $\psi$ to be independent.  The formula of Theorem~\ref{thm:wr} for winding ratios extends to those for $u$ and $J$ on these tori, with $n$ replaced by $\nabla (u\cdot B)/|\nabla (u\cdot B)|^2$.

\section{Relation to standard treatments}

We begin with

\begin{defn}
\label{def:Bqs}
A magnetic field $B$ is {\em quasi-symmetric} if it has a quasi-symmetry $u$ and the associated flux function $\psi$ is global and $d\psi \ne 0$ a.e.  
\end{defn}

Using Definitions \ref{def:integrable} and \ref{def:qs}, if $B$ is quasi-symmetric then FGCM is integrable.

The standard approach to quasi-symmetry, as exemplified by \cite{H}, assumes a magnetohydrostatic field $B$ from the start.  
\begin{defn}
A magnetic field $B$ is {\em magnetohydrostatic} (MHS) if $J \times B = \nabla p$ for some function $p$, where $J = \curl B$.  
\end{defn}
It then assumes a flux function $\psi$, i.e.~a function such that $B\cdot \nabla \psi = 0$ and $\nabla \psi \ne 0$ a.e.  Given the MHS assumption, this is not a great restriction, because $p$ satisfies $B\cdot \nabla p=0$; the only catch is that $dp$ might not be non-zero a.e.~(in particular in the surrounding vacuum region, but also because except for special cases like axisymmetric fields $dp$ has to be zero at all rational flux surfaces \cite{Bo1}).  Then the level sets of $\psi$ are assumed to be bounded in the region of interest and hence 2-tori.  Actually, \cite{H} scales $\psi$ to be the toroidal flux bounded by the level set of $\psi$, but that is not essential.  The standard approach also assumes that $p$ is constant on flux surfaces, which is automatic if the field has density of irrational flux surfaces but might otherwise be a restriction.  Then it is proved that there are ``Boozer'' coordinates  \cite{Bo1} which in particular make the magnetic field lines straight.  Guiding-centre motion is formulated in the Boozer angles as a Lagrangian system and seen to have an ignorable linear combination of the angles if the field strength is constant along a family of straight lines on each flux surface (not in general the same straight lines as the fieldlines), and so guiding-centre motion is integrable.  The field is said to be quasi-symmetric if this is the case.

An alternative approach is due to Hamada \cite{Ha} but requires $dp \ne 0$ a.e.  It constructs a different coordinate system on flux surfaces, but with similar properties, and quasi-symmetry is identified as the result of an ignorable coordinate again.  \cite{H} identifies a whole class of coordinate systems that will do as well.

Here we explain how our approach connects to these.  In our definition and analysis of quasi-symmetry, we have not required the field to be MHS, but if it is MHS and has quasi-symmetry $u$ in our sense with a global flux function $\psi$ and if $p$ is constant on flux surfaces it turns out that $u\cdot B$ is also. The latter in this case will be denoted by $C$. We showed that quasi-symmetry implies that $[u,B]=0$ and $[u, B/|B|^2]=0$.  We claim that the resulting AL coordinates on flux surfaces augmented by $\psi$ give Hamada coordinates in the first case and Boozer coordinates in the second.  Because of $L_u|B|=0$ it follows that $|B|$ is constant along the $u$-lines, which are straight in either case.  So under the assumptions of MHS with $p$ constant on flux surfaces, quasi-symmetry in our sense implies quasi-symmetric in the standard sense.



We will now prove these statements.  Afterwards we will address the converse question.

\begin{thm}
\label{thm:u.B}
If $u$ is a quasi-symmetry for an MHS field $B$ and $p$ is constant on flux surfaces then $u\cdot B$ is constant on flux surfaces.  
\end{thm}

\begin{proof}
We already have $L_u (u\cdot B)=0$ (Theorem~\ref{thm:add}).
Now $L_B (u\cdot B) = L_B i_u B^\flat = i_u L_B B^\flat$ because $[u,B]=0$.
Translated to differential forms, the MHS equation says $i_BdB^\flat = dp$.  So $L_B B^\flat = dp + d|B|^2$.
Thus $L_B (u\cdot B) = i_u dp + i_u d|B|^2$.  The first term is zero because $i_ud\psi=0$ and $p$ is constant on flux surfaces.  The second is zero from $L_u|B|=0$.  $u$ and $B$ are independent and tangent to flux surfaces.  Combining these two results, $u\cdot B$ is constant on flux surfaces.
\end{proof}

Thus, in the MHS case, $u\cdot B$ is a function\footnote{Technically, this is not quite correct, because if a level set of $\psi$ has more than one regular component, then $u\cdot B$ could take different values on the different components, but we use the formulation $u\cdot B=C(\psi)$ with this understanding. Similarly, we will write $p$ being constant on flux surfaces as $p=p(\psi)$.} 
$C(\psi)$ and so together with $B\times u=\nabla\psi$ we can describe the magnetic field in terms of $u$ and $\psi$.  We can do the same for $J$.  The results are stated in the following theorem.

\begin{thm}
\label{thm:qmBJ}
If $u$ is a quasi-symmetry for an MHS field $B$ with $p$ constant on flux surfaces then $B$ has the form
\begin{equation}
B=\frac{1}{|u|^2}(C(\psi)u + u\times \nabla \psi).
\label{eq:Binqs}
\end{equation}
Assume $C$ and $p$ are absolutely continuous functions of $\psi$. Then 
\begin{equation}
J = -p'(\psi) u-C'(\psi) B.  
\label{eq:Jinqs}
\end{equation}
\end{thm}

\begin{proof}
Equation~(\ref{eq:Binqs}) comes straight from $B\times u=\nabla\psi$ if we cross with $u$ and use $u\cdot B = C(\psi)$. Absolute continuity of a function is enough for its derivative to exist a.e. and for the fundamental theorem of calculus to hold. 
The MHS condition implies that $J$ is tangent to flux surfaces, so it is a linear combination of $u$ and $B$, say $J=\kappa u+\lambda B$ for some functions $\kappa,\lambda$ of $\psi$. Then from $\nabla p=J\times B=\kappa u\times B=-\kappa\nabla\psi$ we find $\kappa=-p'$. Likewise from (\ref{eq:LubigBflat}) we have $\nabla C=u\times J=\lambda u\times B=-\lambda\nabla\psi$, and so $\lambda=-C'$.
\end{proof}

\begin{rem}\normalfont
For future reference, it is also useful to express the quasi-symmetry $u$ in terms of $B$ and $\psi$. So now, if we cross $B\times u=\nabla\psi$ with $B$, we have
\begin{equation}
u=\frac{1}{|B|^2}(C(\psi)B - B\times \nabla \psi),
\label{eq:uinms}
\end{equation}
using $u\cdot B=C(\psi)$. An alternative expression can be obtained by crossing $B\times u=\nabla\psi$ with $\nabla|B|$ and using $u\cdot\nabla|B|=0$. In this way, the quasi-symmetry can be written as
\begin{equation}
\label{eq:unoms}
u = \frac{\nabla\psi\times\nabla|B|}{B\cdot\nabla|B|}.
\end{equation}
We emphasise that the two expressions are not equivalent, as the former assumes MHS fields in using $u\cdot B=C(\psi)$, while the latter does not, but uses $L_u|B|=0$ instead.
\end{rem}

\begin{defn}
Given a magnetic field $B$ with a flux function $\psi$, coordinates $(\theta^1,\theta^2,\psi)$, $\theta^j \in \S^1 = \R/2\pi\Z$, are  {\em magnetic coordinates} if the $B$-lines are straight in them.  
\end{defn}
Note that this remains true under any linear transformation on the $\theta^j$ in $SL(2,\Z)$, so it is conventional to take $\theta^1$ in the poloidal direction on a flux surface and $\theta^2$ in the toroidal direction, defined (up to orientation) by its embedding in $\R^3$.  For topologists, a poloidal loop is a ``meridian'' and a toroidal one is a ``longitude''.

\begin{defn}
A set of magnetic coordinates is called {\em Hamada coordinates} if the current density takes the form $J = \nabla I \times \nabla \theta^1 + \nabla G \times \nabla \theta^2$ for some functions $I,G$ of $\psi$ only (in differential forms, the current flux-form $j = i_J\Omega = dI\wedge d\theta^1+ dG \wedge d\theta^2$).
\end{defn}
\begin{defn}
A set of magnetic coordinates is called {\em Boozer coordinates} if the magnetic field takes the form $B = I \nabla \theta^1 + G\nabla \theta^2 + K\nabla \psi$ with $I,G$ functions of $\psi$ only (in differential forms, $B^\flat = I d\theta^1 + G d\theta^2 + K d\psi$).
\end{defn}
Note that in each case there is freedom to choose where to put the origin of $(\theta^1,\theta^2)$ on each flux surface.

\begin{thm}
\label{thm:Hamada}
If $u$ is a quasi-symmetry for an MHS field $B$ then AL coordinates for $[u,B]=0$ give Hamada coordinates and $|B|$ is constant along a system of straight lines in these coordinates.
\end{thm}

\begin{proof}
Quasi-symmetry implies $[u,B]=0$ (Theorem~\ref{thm:conseq}).  Construct AL coordinates $(\theta^1,\theta^2)$ for this commutation relation and augment to 3D by $\psi$.  Then both $u$ and $B$ are constant vector fields on each flux surface in these coordinates, i.e.~$u = u^1\partial_{1}+u^2\partial_{2}$ and $B = B^1\partial_{1}+B^2\partial_{2}$, with the coefficients being functions of $\psi$ only, where $\partial_i$ is short for $\partial_{\theta^i}$.  In particular, the $B$-lines are straight, so $(\theta^1,\theta^2)$ are magnetic coordinates.  

Under the MHS condition with $p$ constant on flux surfaces, $J$ is given by (\ref{eq:Jinqs}).
Thus the current 2-form $j = i_J\Omega = -p' i_u\Omega  - C' i_B\Omega$.
Now $\Omega$ can be written as $\Omega = \mathcal{J} d\psi \wedge d\theta^1 \wedge d\theta^2$ where $\mathcal{J}$ is the Jacobian of the coordinate system.  The Jacobian in these coordinates is a function of $\psi$ only, because $\div u = \div B = 0$ implies $\div \partial_i = 0$ for $i=1,2$, so $0=di_{\partial_1}\Omega = -d\mathcal{J} \wedge d\psi \wedge d\theta^2$ and $0=di_{\partial_2}\Omega = d\mathcal{J}\wedge d\psi \wedge d\theta^1$.  Thus $d\mathcal{J}\wedge d\psi = 0$ and hence $\mathcal{J}$ is a function of $\psi$.
Finally, 
\begin{equation}
j=\mathcal{J} (-p'(u^2d\psi \wedge d\theta^1 - u^1d\psi \wedge d\theta^2) - C'(B^2d\psi \wedge d\theta^1 - B^1d\psi \wedge d\theta^2)),
\end{equation}
so has the form $j_1 d\psi \wedge d\theta^1 + j_2 d\psi \wedge d\theta^2$ for some $j_i(\psi)$, which (together with being magnetic coordinates) is the defining condition of Hamada coordinates.  Because $u$ is constant in these coordinates on a flux surface, it follows from $L_u|B|=0$ that $|B|$ is constant along a system of straight lines.
\end{proof}

Note that $J$ is also constant in Hamada coordinates on each flux surface, because we proved in (\ref{eq:Jinqs}) it is a linear combination (by functions of $\psi$) of $u$ and $B$, which are constant on each flux surface.  Thus instead of making AL coordinates for $[u,B]=0$, we could equally well make AL coordinates for $[J,B]=0$ or for $[J,u]=0$, the first of which holds for any MHS field and the second of which holds for any quasi-symmetric field.  To prove that $[J,B]=0$ for any MHS field, note simply that 
\begin{equation}
i_{[J,B]}\Omega = i_J L_B\Omega - L_B i_J\Omega = -L_BdB^\flat = -dL_B B^\flat = 0 ,
\end{equation}
because $L_B B^\flat = d(p+|B|^2)$.  $[J,B]=0$ gives Hamada coordinates and $|B|$ is constant along a system of straight lines if the field is quasisymmetric.  

\begin{thm}
\label{thm:Boozer}
If $u$ is a quasi-symmetry for an MHS field $B$ then AL coordinates for $[u,B/|B|^2]=0$ give Boozer coordinates and $|B|$ is constant along a system of straight lines in these coordinates.
\end{thm}

\begin{proof}
Quasi-symmetry implies $[u,B/|B|^2]=0$ (Theorem~\ref{thm:add}).  Construct AL coordinates $(\theta^1, \theta^2)$ for this commutation relation and augment by $\psi$.  Then both $u$ and $B/|B|^2$ are constant vector fields on each flux surface in these coordinates.  In particular the $B$-lines are straight, even if $|B|$ might vary along them.  So these are magnetic coordinates again.  

Under the MHS condition with $p$ constant on flux surfaces, then we see the AL coordinates have the additional property that lines of $\nabla \psi \times B$ are straight too.  This is because $u\cdot B$ is constant on flux surfaces, so $\nabla \psi \times B /|B|^2 = u - C B/|B|^2$ (coming from (\ref{eq:uinms})) is a constant vector field on each flux surface.  This is a property of Boozer coordinates, but
the defining property of Boozer coordinates (beyond being magnetic coordinates) is that $B^\flat = I d\theta^1 + G d\theta^2 + K d\psi$ for some functions $I,G,K$ with $I,G$ functions of $\psi$ only.  So now we prove this holds for our AL coordinates.  Without the constraints on $I$ and $G$, $B^\flat$ has the above form because the coordinate 1-forms form a basis for the cotangent space.  As before, we proved $C=u\cdot B$ is constant on flux surfaces, so $C=i_uB^\flat = I u^1 + G u^2$.  Also $1= i_{B/|B|^2} B^\flat = I h^1 + G h^2$, where $h^j = B^j/|B|^2$.  The coefficients $(u^j,h^j)$ in these two equations are functions of $\psi$ only, so it follows that $I$ and $G$ are too.  Thus the AL coordinates for $[u,B/|B|^2]=0$ in an MHS quasi-symmetric field are Boozer coordinates.  Because $u$ is constant in these coordinates and $|B|$ is constant along $u$-lines, it follows that $|B|$ is constant along a system of straight lines on each flux surface.
\end{proof}

Again, we note that the same coordinates are obtained if we start from the commutation relation $[u,\nabla\psi \times B/|B|^2]=0$, which holds in quasisymmetry, or from $[B/|B|^2,\nabla\psi \times B/|B|^2]=0$, whose significance comes from the next result.

\begin{thm}
\label{thm:Boozercomm}
If $B$ is MHS with flux function $\psi$ and $p$ constant on flux surfaces, then $[B/|B|^2,\nabla\psi \times B/|B|^2]=0$.  
\end{thm}

\begin{proof}
Write $h = B/|B|^2$ and $k=\nabla \psi \times h$. We apply
\begin{equation}
i_{[h,k]} = L_h i_k - i_k L_h
\end{equation}
to $d\psi,B^\flat,k^\flat$ in turn, using the operator
\begin{equation}
L_h = |B|^{-2} L_B + d|B|^{-2} \wedge i_B,
\end{equation}
which reduces to $|B|^{-2}L_B$ on $0$-forms.
In the first case we get immediately $i_{[h,k]}d\psi=0$, since $L_hd\psi=0$ and $i_k d\psi = 0$. 
Next, inserting the MHS condition in the form $L_B B^\flat = d(p+|B|^2)$, we have
\begin{equation}
i_{[h,k]}B^\flat=-i_kL_hB^\flat=-i_k(|B|^{-2}d|B|^2+|B|^2d|B|^{-2})=0,
\end{equation}
as well. For the last case, note first that
\begin{equation}
i_{[h,k]}k^\flat=L_h|k|^2-i_kL_hk^\flat=|B|^{-2} L_Bi_kk^\flat-|B|^{-2}i_kL_Bk^\flat=|B|^{-2}i_{[B,k]}k^\flat,
\end{equation}
since $i_Bk^\flat=0$.
But, using $i_k\Omega = d\psi \wedge h^\flat$, we also see that
\begin{equation}
i_Bi_{[B,k]}\Omega=i_B(L_Bi_k\Omega-i_kL_B\Omega)=i_B(d\psi\wedge L_Bh^\flat)=-d\psi\wedge L_B1=0,
\end{equation}
since $L_B\psi=0$, $L_B\Omega=0$ and $i_Bd\psi=0$. Therefore the vector field $[B,k]$ is parallel to $B$, and so $i_{[h,k]}k^\flat=0$ too.  Since $(\nabla\psi,B,k)$ form a basis for the tangent space, we deduce that $[h,k]=0$.
\end{proof}

Thus, starting with an MHS field $B$ with a flux function $\psi$ and $p$ constant on flux surfaces, one can construct Boozer coordinates as AL coordinates for the commutation relation of Theorem~\ref{thm:Boozercomm}.  If $B$ is in addition quasisymmetric then $|B|$ is constant along a system of straight lines, namely the $u$-lines, because $L_u|B|=0$ and $u$, given by (\ref{eq:uinms}), is a constant vector field on each flux surface.

The standard approaches to quasi-symmetry are based on a symmetry (ignorable coordinate) of the gyro-averaged Lagrangian in Boozer or Hamada coordinates.  The usual presentations are flawed, however, because the gyro-averaged Lagrangian
\begin{equation}
L(X,\dot{X})=\frac12 m (b\cdot \dot{X})^2 + eA\cdot \dot{X} -\mu|B(X)|
\end{equation}
does not define a Lagrangian system in the standard sense, neither on $(X,\dot{X}) \in T\R^3$ nor on $(\theta,\dot{\theta}) \in T(\S^1\times \S^1)$, because it is degenerate (the second derivative of $L$ with respect to the three components of $\dot{X}$ has rank 1).  The good way to deal with this \cite{L} is to write the variational principle in Arnol'd-Cartan form \cite{A} in extended state space $(X,v_\pl,t)$ as $\delta  \int_\gamma \alpha = 0$ over variations of compact support of paths $\gamma$ with 
\begin{equation}
\alpha = mv_\pl b^\flat + eA^\flat - (\frac12 mv_\pl^2 + \mu|B(X)|)\ dt .
\end{equation}
The vector field $V$ resulting from such a variational principle is given by the unique choice $(V,1)$ in the kernel of $d\alpha$ (which is one-dimensional because the extended state space is odd-dimensional and $B$ is assumed non-zero).
Then a continuous symmetry is defined to be a vector field $U$ such that $L_U \alpha = dZ$ for some function $Z$.  This is because flowing with such a vector field does not change the variational principle (actually, one could replace $dZ$ by any closed 1-form, but it is conventional to take it exact).  It follows, as in the standard Noether theorem, that $i_Ud\alpha = d(Z-i_U\alpha)$ and so
$i_{(V,1)}d(Z-i_U\alpha) = -i_U i_{(V,1)}d\alpha = 0$, so $K=Z-i_U\alpha$ is conserved by $(V,1)$.
If $U=(u,0)$ then 
\begin{equation}
i_U\alpha = e i_u A^\flat + mv_\pl i_u b^\flat = e u\cdot A + mv_\pl u\cdot b.
\end{equation}

Finally, we address the converse question:~given a quasi-symmetric system in the standard sense, identify the quasi-symmetry in our sense.

\begin{thm}
If magnetic field $B$ has a flux function $\psi$ and is MHS with $p$ constant on flux surfaces and density of irrational surfaces and $p'(\psi)\ne 0$ a.e., and $|B|$ is constant along a family of straight lines in Hamada coordinates, then it is quasisymmetric with $u=-(J+C'B)/p'$ for a function $C$ of $\psi$ such that $C=u\cdot B$.
\end{thm}

\begin{proof}
For an MHS field $[J,B]=0$.  If $dp \ne 0$ then $J,B$ are independent.  Take AL coordinates for this commutation relation.  By the discussion after Theorem~\ref{thm:Hamada}, they are Hamada coordinates.  If $|B|$ is constant along a family of straight lines then that implies $|B|$ is constant along 
\begin{equation}
u=\kappa J + \lambda B
\label{eq:ukappa}
\end{equation}
for some functions $\kappa, \lambda$ of $p$ or equivalently of $\psi$ if a different flux function has been chosen with the same level sets.  Their ratio is determined by the lines of constant $|B|$, but their magnitudes are otherwise free.

Now $i_ui_B\Omega = \kappa i_J i_B\Omega = -\kappa dp = -\kappa p'd\psi$.  Therefore $L_u\beta=0$.  Choose, in particular, $\kappa = -1/p'$ to obtain $i_ui_B\Omega = d\psi$.  


It remains to prove that $L_uB^\flat=0$. From (\ref{eq:ukappa}), $i_ui_J\Omega =\lambda i_Bi_J\Omega = \lambda dp = \lambda p'd\psi$, by the MHS equation. And $[u,B]=0$ since $[J,B]=0$.  From the MHS condition again $L_BB^\flat=d(p+|B|^2)$, and $[u,B]=0$, we have $L_B(u\cdot B) = i_uL_B B^\flat = i_u d(p + |B|^2) = 0$, because both $p$ and $|B|$ are constant along $u$. Then density of irrational surfaces implies, assuming continuity, that $u\cdot B$ is constant on flux surfaces. In other words, $u\cdot B$ is a function $C$ of $\psi$. Therefore
\begin{equation}
L_uB^\flat = i_udB^\flat + di_uB^\flat = i_ui_J\Omega + d(u\cdot B) = (\lambda p' + C')d\psi.
\label{eq:LuBfH} 
\end{equation}
Thus, choosing $\lambda = -C'/p'$, we satisfy the last condition for quasisymmetry. 
\end{proof} 

\begin{rem}\normalfont
In the previous proof, we can show that $L_uB^\flat=0$ using circle averaging instead, as follows. First we note that if $|B|$ is not constant on a flux surface then the $u$-lines are closed.  If exceptionally, $|B|$ is constant on a flux surface then the ratio is undetermined, but we can choose it to make the $u$-lines closed. The period $\tau$ of the $u$-lines is constant on flux surfaces but in general varies with $\psi$.  We are free to simultaneously scale $\kappa$ and $\lambda$ by any function of $\psi$, however.  Thus we can scale them to make $\tau=2\pi$.  With this choice we can now apply circle-averaging (\ref{eq:circleavg}) to (\ref{eq:LuBfH}). The average of the lefthand side is zero, being $L_u$ of something.  All of $\lambda, p$ and $C$ are constant along $u$.  Thus the average of the right hand side is just itself. Consequently $0 = (\lambda p' + C')d\psi$. It follows that $L_uB^\flat = 0$.
\end{rem}

\begin{thm}
If magnetic field $B$ has a flux function $\psi$ and is MHS with $p$ constant on flux surfaces and density of irrational surfaces, and $|B|$ is constant along a family of straight lines in Boozer coordinates, then it is quasisymmetric with $u = (CB + \nabla \psi \times B)/|B|^2$ for a function $C$ of $\psi$ such that $(C^2/2)'=-|u|^2p'-u\cdot J$. 
\end{thm}

\begin{proof}
For an MHS field with a flux function $\psi$ we proved that $[B/|B|^2, \nabla \psi \times B/|B|^2]=0$ (Theorem~\ref{thm:Boozercomm}), and the corresponding AL coordinates are Boozer (discussion after Theorem~\ref{thm:Boozercomm}).  If $|B|$ is constant along a family of straight lines in these coordinates then there are functions $C, \lambda$ of $\psi$ only such that $|B|$ is constant along 
\begin{equation}
u = (CB + \lambda \nabla \psi \times B)/|B|^2. 
\label{eq:Clambda}
\end{equation}
The ratio of $C,\lambda$ is determined by the lines of constant $|B|$, but their magnitudes are otherwise free.  We see that $C=u\cdot B$, hence $u\cdot B$ is constant on flux surfaces. 

Now $B\times u = \lambda \nabla \psi$, so we obtain $i_ui_B\Omega = \lambda d\psi$ and $L_u \beta = 0$ accordingly. Let us choose $\lambda=1$ to obtain $i_ui_B\Omega=d\psi$.


It remains to prove that $L_uB^\flat=0$. The MHS equation can be written as $i_Bi_J\Omega=dp$. Thus, from the above expression for $u$, we deduce that
\begin{equation}
i_ui_J\Omega = \frac{1}{|B|^2}(Ci_Bi_J\Omega-i_J(d\psi\wedge B^\flat) = \frac{1}{|B|^2}(Cdp+(i_JB^\flat)d\psi) = \kappa d\psi
\label{eq:iuiJ}
\end{equation}
since $i_Jd\psi=0$, with 
\begin{equation}
\kappa = (Cp'+ J\cdot B)/|B|^2 = (|u|^2p'+ u\cdot J)/C. 
\label{eq:alpha}
\end{equation}
Therefore
\begin{equation}
\label{eq:LuBfB}
L_uB^\flat = i_ui_J\Omega+d(u\cdot B)=(\kappa+C')d\psi.
\end{equation}
Next note that $[u,B/|B|^2]=0$ because $u$ is given by (\ref{eq:Clambda}), $B\cdot\nabla C=0$ and 
$[B/|B|^2, \nabla \psi \times B/|B|^2]=0$. But $|B|$ is constant along $u$, so that implies $[u,B]=0$.  Apply then $L_B$ to (\ref{eq:iuiJ}). We have $[u,B]=0$, $[J,B]=0$ from the MHS condition, and $L_B\Omega=0$.  Hence $L_B\kappa=0$. Using density of irrational surfaces it follows, assuming continuity, that $\kappa$ is constant on flux surfaces, i.e., $\kappa=\kappa(\psi)$.
Choose then $\kappa=-C'$ to obtain $L_uB^\flat=0$.
Inserting this in (\ref{eq:alpha}) proves the result for $(C^2/2)'$.
\end{proof}

\begin{rem}\normalfont
Alternatively in the previous proof we can show that $L_uB^\flat=0$ using circle averaging, as follows. If $|B|$ is not constant on a flux surface then the $u$-lines are closed.  If it is constant then we can choose the ratio $C:\lambda$ to make the $u$-lines closed.  In either case, the period of the $u$-lines is constant on a flux surface.  We can scale $C,\lambda$ simultaneously by a function of $\psi$ to make the period $2\pi$. Apply circle averaging to (\ref{eq:LuBfB}) and use $\kappa, C$ constant along $u$ to obtain $0=(\kappa+C')d\psi$. Thus $L_uB^\flat=0$.
\end{rem}

Another common treatment of quasi-symmetry for an MHS field with flux function and $p$ constant on flux surfaces (e.g. \cite{SH}), is based on the relation
\begin{equation}
\label{eq:He1}
B\times\nabla\psi = E\,\nabla\psi\times\nabla|B| + FB,
\end{equation}
where $E=-|B|^2/B\cdot\nabla|B|$ and $F=B \times \nabla \psi \cdot \nabla |B|/(B\cdot\nabla|B|)$. Quasi-symmetry in this approach is then formulated as $F$ being constant on flux surfaces,
\begin{equation}
\label{eq:He2}
F = F(\psi).
\end{equation}
This setup fits in our framework because equation (\ref{eq:He1}) is none other than $i_ui_B\Omega=d\psi$ and $L_u|B|=0$ combined together, and condition (\ref{eq:He2}) says that $u\cdot B$ is constant on flux surfaces in the MHS case. To see the first one, cross $B\times u=\nabla\psi$ with $B$ to obtain $B\times \nabla \psi = (u\cdot B)B - |B|^2 u$, and insert (\ref{eq:unoms}) for $u$. The second one follows from Theorem~\ref{thm:u.B}, since the function $F$ is precisely $u\cdot B$, i.e., $F=C$ for MHS fields.

Lastly we show that $L_ub^\flat=0$ implies $\int_c dl$ is constant when the curve $c$ is drawn from any continuous family of field line segments within a given flux surface and with fixed endpoint values of $|B|$. Let $\gamma:[s_0,s_1]\rightarrow \mathbb{R}^3$ be the restriction of a field line to an interval $[s_0,s_1]$ such that $|B|(\gamma(s_0)) = k_0$ and $|B|(\gamma(s_1)) = k_1$, where $k_0,k_1\in\mathbb{R}_+$. If $\phi_\lambda$ is the $u$-flow then $\gamma_\lambda = \phi_\lambda\circ \gamma$ is a field line segment contained in the same flux surface as $\gamma$ for each $\lambda$. In addition the integral $I_\lambda = \int_{\gamma_\lambda}dl$ is independent of $\lambda$ because
\begin{align}
\frac{d}{d\lambda}\int_{\phi_\lambda\circ c} dl = \frac{d}{d\lambda}\int_{\phi_\lambda\circ c} b^\flat = \frac{d}{d\lambda}\int_{ c} \phi_\lambda^* b^\flat = \int_c L_ub^\flat = 0.
\end{align}
Therefore the arc lengths of the field line segments $\gamma_\lambda$ are all the same. Moreover because $L_u|B| = 0$ the endpoint values of $|B|$ for $\gamma_\lambda$ are independent of $\lambda$, i.e. $|B|(\gamma_\lambda(s_0)) = k_0$ and $|B|(\gamma_\lambda(s_1)) = k_1$ for each $\lambda$. The desired result now follows upon noting that any continuous family of field line segments in a given flux surface with fixed endpoint values of $|B|$ can be generated by flowing some field line segment along $u$.

\section{Quasi-symmetric Grad-Shafranov equation}

In the axisymmetric case, magnetohydrostatics is reduced to a nonlinear elliptic partial differential equation called the Grad-Shafranov\footnote{but previous versions were published by L\"ust \& Schl\"uter, 1957 and by Chandrasekhar \& Prendergast, 1956, and in the fluids context the analogous equation was published by Hicks in 1899} 
(GS) equation \cite{GR, Sh}.  It takes as input two functions $p(\psi)$ and $C(\psi)$ and is an equation for $\psi(r,z)$ in cylindrical polar coordinates.  The GS equation has a nice variational principle \cite{BB}\footnote{but probably there are earlier references}, reasonable existence theory for solutions \cite{BB,AM} and well developed codes for its numerical solution, e.g.~Ch.4 of \cite{Ja}.

Here we generalise the GS equation to magnetohydrostatics with a general quasi-symmetry $u$.   

First we derive a pre-GS equation which does not assume magnetohydrostatics.  Furthermore, the only part of quasi-symmetry that it uses is $i_ui_B\Omega = d\psi$.

\begin{thm}
If $i_ui_B\Omega=d\psi$ then 
\begin{equation}
\Delta \psi - \frac{u\times v}{|u|^2}\cdot \nabla \psi + \frac{u\cdot v}{|u|^2}\,u\cdot B - u\cdot J = 0,
\label{eq:preGS}
\end{equation}
where $\Delta = \div\nabla$, $v=\curl u$ and $J= \curl B$.
\end{thm}

\begin{proof}
We have $B\times u = \nabla \psi$ and so
\begin{equation}
B^\flat \wedge u^\flat = i_{B\times u}\Omega = i_{\nabla\psi}\Omega.
\label{eq:Bwedgeu}
\end{equation}
Applying $d$ to the above equation, we obtain
\begin{equation}
\Delta \psi\ \Omega=d(B^\flat \wedge u^\flat) = dB^\flat \wedge u^\flat - B^\flat \wedge du^\flat = i_J\Omega \wedge u^\flat - B^\flat \wedge i_v\Omega
= (u\cdot J - B\cdot v)\,\Omega.
\end{equation}
Hence by non-degeneracy of $\Omega$,
\begin{equation}
\label{eq:preGS0}
\Delta\psi = u\cdot J - B\cdot v.
\end{equation}
For the last term of (\ref{eq:preGS0}), contract (\ref{eq:Bwedgeu}) with $v$ and $u$ to find
\begin{equation}
v\times u\cdot\nabla\psi = i_ui_v(B^\flat \wedge u^\flat) = i_u((B\cdot v)u^\flat-(u\cdot v)B^\flat) = (B\cdot v)|u|^2-(u\cdot v)(u\cdot B).
\end{equation}
Hence
\begin{equation}
B\cdot v = \left(v\times u \cdot \nabla \psi + (u\cdot v)(u\cdot B)\right)/|u|^2.
\end{equation}
Substituting this into (\ref{eq:preGS0}), we deduce equation (\ref{eq:preGS}).
\end{proof}

An immediate consequence of $i_ui_B\Omega = d\psi$ is also the additional equation $u\cdot \nabla \psi = 0$, which comes by contracting with $u$.

Next we derive two further conditions which follow from $\div B=0$ and $L_uB^\flat=0$, assuming in addition the properties $\div u=0$ and $L_u(u\cdot B)=0$ of a quasi-symmetry.

\begin{thm}
\label{thm:extra1}
If $i_ui_B\Omega=d\psi$, $\div u=0$ and $L_u(u\cdot B)=0$, then $\div B = 0$ iff $B\cdot w=0$ where $w^\flat=L_uu^\flat$.
\end{thm}

\begin{proof}
Recalling $\beta=i_B\Omega$, we have
\begin{align}
|u|^2d\beta &= (i_ud\beta)\wedge u^\flat = (L_u\beta-di_u\beta)\wedge u^\flat = (i_{[u,B]}\Omega + i_BL_u\Omega)\wedge u^\flat = (i_{[u,B]}u^\flat)\,\Omega\\
& = (L_ui_Bu^\flat - i_BL_uu^\flat)\,\Omega = - (B\cdot w)\,\Omega,
\end{align}
because $di_u\beta$, $L_u\Omega$ and $L_ui_Bu^\flat$ are zero. $d\beta = (\div B)\,\Omega$, hence the result.
\end{proof}

\begin{thm}
\label{thm:extra2}
Under the assumptions of the previous theorem, $L_uB^\flat=0$ iff $B\times w = [u,\nabla\psi]$. 
\end{thm}

\begin{proof}
Note first that since $i_uL_uB^\flat=L_ui_uB^\flat=0$, we have $i_u(L_uB^\flat\wedge u^\flat)=-|u|^2L_uB^\flat$. Thus, $L_uB^\flat$ vanishes iff $L_uB^\flat\wedge u^\flat$ does. Applying then $L_u$ to (\ref{eq:Bwedgeu}), we obtain
\begin{equation}
(L_uB^\flat)\wedge u^\flat = L_ui_{\nabla\psi}\Omega - B^\flat\wedge L_uu^\flat = i_{[u,\nabla\psi]}\Omega - i_{B\times w}\Omega,
\end{equation}
since $L_u\Omega=0$. $\Omega$ is non-degenerate, hence the result.
\end{proof}

Thus the pre-GS equation requires extra conditions to guarantee that $\div B=0$ and $L_u B^\flat=0$. Both of them are automatic in the case of isometries, that is, if $u$ is a Killing field. To see this, write $w^\flat=L_uu^\flat=i_uL_ug$ and $[u,\nabla\psi]^\flat=L_u(d\psi)-i_{\nabla\psi}L_ug$, recalling Lemma \ref{lem:1}, and take into account the first condition, $L_u\psi=0$. In the general case, however, they appear to be non-trivial additional conditions.

The condition of Theorem \ref{thm:extra1} is also automatic if $w=0$. For a quasi-symmetry, the latter was precisely the condition for $u$ to be a Killing field (Theorem \ref{thm:killingqs}). The next result shows that this also true if instead $u$ satisfies the first and third supplementary conditions.

\begin{thm}
\label{thm:uKf}
If $\div u=0$, $L_u\psi=0$ and $B\times w = [u,\nabla\psi]$, then $w=0$ iff $L_ug=0$.
\end{thm}

\begin{proof}
If $L_ug=0$, then straightforwardly $w^\flat=L_uu^\flat=i_uL_ug=0$.

Let $w=0$. Then $[u,\nabla\psi]=0$. Using $L_u\Omega=0$ and $L_ud\psi=dL_u\psi=0$,
\begin{equation}
i_{[u,u\times\nabla\psi]}\Omega = L_ui_{u\times\nabla\psi}\Omega-i_{u\times\nabla\psi}L_u\Omega = L_u(u^\flat\wedge d\psi) = L_uu^\flat\wedge d\psi = 0.
\end{equation}
Thus, $[u,u\times\nabla\psi]=0$ too, since $\Omega$ is non-degenerate.

Consider then the basis $(u,\nabla\psi,u\times\nabla\psi)$. $i_uL_ug=L_uu^\flat=0$. Furthermore, since $u$ commutes with $\nabla\psi$ and $u\times\nabla\psi$,
\begin{align}
i_{\nabla\psi}L_ug &= L_ui_{\nabla\psi}g = L_ud\psi = dL_u\psi = 0,\\
i_{u\times\nabla\psi}L_ug & = L_ui_{u\times\nabla\psi}g = L_u(u\times\nabla\psi)^\flat = L_ui_{\nabla\psi}i_u\Omega = i_{\nabla\psi}i_uL_u\Omega = 0,
\end{align}
using $L_u\psi=0$ and $L_u\Omega=0$ again. Hence $L_ug=0$.
\end{proof}

We now turn to combining quasi-symmetry with magnetohydrostatics, obtaining another of our main results.

\begin{thm}\label{qsGS}
If MHS field $B$ is quasi-symmetric with quasi-symmetry $u$, flux function $\psi$, and $p$ a function of $\psi$, then
\begin{equation}
\Delta\psi -\frac{u\times v}{|u|^2}\cdot \nabla \psi+ \frac{u\cdot v}{|u|^2}\,C(\psi) + CC'(\psi) + |u|^2 p' (\psi)= 0,
\label{eq:qsGS}
\end{equation}
where $v = \curl u$ and $C=u\cdot B$.
\end{thm}

\begin{proof}
In the MHS case with $p$ constant on flux surfaces, $u\cdot B$ is a function $C(\psi)$ from Theorem \ref{thm:u.B}, and $u\cdot J=-CC'-|u|^2p'$ from (\ref{eq:Jinqs}). Put these into the pre-GS equation (\ref{eq:preGS}) to obtain (\ref{eq:qsGS}).  
\end{proof}

Equation (\ref{eq:qsGS}) is our quasi-symmetric Grad-Shafranov equation.  For given $u$, $C(\psi)$, and $p(\psi)$ it comprises a semilinear elliptic PDE for the dependent variable $\psi$. 
Solutions of (\ref{eq:qsGS}), however, do not necessarily give MHS fields.  There are several additional conditions that are required.

First of all, equation (\ref{eq:qsGS}) needs supplementing by the condition $L_u\psi=0$, equivalently
\begin{equation}
\label{eq:extra0}
u\cdot \nabla \psi = 0,
\end{equation}
which says that $u$ leaves invariant the solutions $\psi$ of (\ref{eq:qsGS}) and reduces it effectively to 2D. 
For the special case of axisymmetry, $u = r\hat{\phi}$, then $|u|=r$, $v = 2\hat{z}$, $u\cdot v = 0$ and $u\times v = 2r \hat{r}$, so the usual GS equation is recovered in cylindrical polar coordinates. Similarly, for the case of helical symmetry, $u = r\hat{\phi}+l\hat{z}$, where $l$ is a constant, then $|u| = \sqrt{r^2+l^2}$, $v=2\hat{z}$, $u\cdot v = 2l$ and $u\times v = 2r\hat{r}$, so the helical GS equation \cite{JOKF} is obtained.

Recalling Theorem \ref{thm:qmBJ}, the magnetic field $B$ can be obtained by formula (\ref{eq:Binqs}).  If $w=0$, as for axisymmetry and helical symmetry, no further conditions beyond (\ref{eq:qsGS}) and (\ref{eq:extra0}) are required for MHS fields.

If $w\neq0$, however, it is not automatic from (\ref{eq:Binqs}) that $\div B = 0$ nor that $L_uB^\flat=0$. Thus, in general, one must add the conditions of Theorems \ref{thm:extra1} and \ref{thm:extra2} to ensure them.  The first one reads
\begin{equation}
\label{eq:extra1}
(u\times w)\cdot\nabla\psi - (u\cdot w) C(\psi)=0,
\end{equation}
as we can see by replacing $B$ from (\ref{eq:Binqs}) into $B\cdot w=0$. For the non-isometry case, the second one can be reduced to
\begin{align}
\label{eq:extra2}
\left[v\times w - 2(w\cdot\nabla)u\right]\cdot\nabla\psi&=0, \\
\label{eq:extra3}
\left[(u\cdot v)w + 2((u\times w)\cdot\nabla)u\right]\cdot\nabla\psi + |w|^2C(\psi)&=0,
\end{align}
as the next result shows. It is worth noting that, owing to $[u,\nabla\psi]$, the original condition in this case involves second-order partial differential equations, but they can be reduced to first-order ones making use of (\ref{eq:extra0}), as described in (\ref{eq:redcom}) below. Thus, in the end, the second-order quasi-symmetric Grad-Shafranov equation is augmented by four first-order quasilinear partial differential equations given by (\ref{eq:extra0})-(\ref{eq:extra3}).
\begin{prop}
Let $B$ be of the form (\ref{eq:Binqs}). If $u$ is not locally a Killing field, then $B\times w = [u,\nabla\psi]$ can be reduced to (\ref{eq:extra2})-(\ref{eq:extra3}) under (\ref{eq:extra0})-(\ref{eq:extra1}).
\end{prop}
\begin{proof}
First of all for any vector field $X$ we have
\begin{align}
\label{eq:redcom}
\nonumber[u,\nabla\psi]\cdot X &= i_{[u,\nabla\psi]}X^\flat = (L_ui_{\nabla\psi}-i_{\nabla\psi}L_u)X^\flat = L_ui_Xd\psi - i_{\nabla\psi}L_uX^\flat\\
& = (i_{[u,X]} + i_XL_u)d\psi - i_{\nabla\psi}L_uX^\flat = i_{\nabla\psi}([u,X]^\flat - L_uX^\flat),
\end{align}
using $L_ud\psi=dL_u\psi=0$ from (\ref{eq:extra0}). Moreover, switching to vector calculus,
\begin{align}
\label{eq:com2con}
\nonumber[u,X]^\flat - L_uX^\flat&=[(u\cdot\nabla)X-(X\cdot\nabla)u]^\flat - i_udX^\flat - d(i_uX^\flat)\\
\nonumber&=[(u\cdot\nabla)X-(X\cdot\nabla)u + u\times\curl X - \nabla(u\cdot X)]^\flat\\
&=[v\times X-2(X\cdot\nabla)u]^\flat.
\end{align}

Now, if $u\times w=0$ then $u\cdot w=0$ from (\ref{eq:extra1}) and so $w=0$, hence $L_ug=0$ from Theorem \ref{thm:uKf}. Therefore, if $u$ is not locally a Killing field, $(u,w,u\times w)$ is a basis. Project then $B\times w = [u,\nabla\psi]$ to the directions $X=u,w,u\times w$, and use (\ref{eq:redcom})-(\ref{eq:com2con}).

For $X=u$, we see directly from (\ref{eq:redcom}) that the projection of $B\times w = [u,\nabla\psi]$ to $u$ is trivially satisfied,
\begin{equation}
0 = u\cdot(B\times w - [u,\nabla\psi]) = w\cdot(u\times B) + i_{\nabla\psi}L_uu^\flat = -\,w\cdot\nabla\psi + w\cdot\nabla\psi.
\end{equation}

For $X=w$, we obtain (\ref{eq:extra2}),
\begin{equation}
0 = w\cdot(B\times w - [u,\nabla\psi]) = -\,i_{\nabla\psi}[v\times w-2(w\cdot\nabla)u]^\flat.
\end{equation}

For $X=u\times w$, using (\ref{eq:extra0}) and (\ref{eq:extra1}) in its original form $B\cdot w=0$, we arrive at (\ref{eq:extra3}),
\begin{align}
\nonumber0 &= u\times w\cdot(B\times w - [u,\nabla\psi]) \\
\nonumber&= [(B\times w)\times u]\cdot w-\,i_{\nabla\psi}[v\times(u\times w)-2((u\times w)\cdot\nabla)u]^\flat\\
\nonumber&= [Cw - (w\cdot u)B]\cdot w-\,i_{\nabla\psi}[(v\cdot w)u-(v\cdot u)w-2((u\times w)\cdot\nabla)u]^\flat\\
&= C|w|^2 + \,i_{\nabla\psi}[(v\cdot u)w+2((u\times w)\cdot\nabla)u]^\flat.
\end{align}
\end{proof}

Writing $\mathcal{V}_1=u\times w-2(w\cdot\nabla)u$ then (\ref{eq:extra2}) becomes $\mathcal{V}_1\cdot\nabla\psi=0$. Adding $|w|^2$ times (\ref{eq:extra1}) to $u\cdot w$ times (\ref{eq:extra3}), we obtain $\mathcal{V}_2\cdot\nabla\psi=0$ with $\mathcal{V}_2=|w|^2u\times w+(u\cdot w)[(u\cdot v) w+2(u\times w\cdot\nabla)u]$. Combining the two equations with $u\cdot\nabla\psi=0$ from (\ref{eq:extra0}) we deduce a further requirement.
\begin{thm}
\label{thm:MHSqscon}
For an MHS field with $p$ a function of $\psi$ to be quasi-symmetric, the quasi-symmetry $u$ must satisfy $u\times\mathcal{V}_1\cdot\mathcal{V}_2=0$.
\end{thm}
\begin{proof}
By Definition \ref{def:Bqs} of quasi-symmetric, $d\psi\neq0$ a.e. Since $u\cdot\nabla\psi=0$, $\mathcal{V}_1\cdot\nabla\psi=0$ and $\mathcal{V}_2\cdot\nabla\psi=0$, it follows that $u,\mathcal{V}_1,\mathcal{V}_2$ must be linearly dependent a.e.
\end{proof}
The condition of the above theorem merits analysis, but there are also the further restrictions given by (\ref{eq:extra1}) and (\ref{eq:qsGS}) itself. And it would be good to also remove the MHS condition.

Finally, the current density $J$ can be found from (\ref{eq:Jinqs}). But $J=\curl B$ is not automatic either. Still for solutions of (\ref{eq:preGS}) and therefore (\ref{eq:qsGS}) this amounts to $L_uB^\flat=0$ again as the next result shows.

\begin{thm}
Let $B,J$ be of the form (\ref{eq:Binqs})-(\ref{eq:Jinqs}). On the set of solutions of the pre-GS equation, $L_uB^\flat=0$ iff $J=\curl B$.
\end{thm}

\begin{proof}
Write $L_uB^\flat=i_udB^\flat+d(i_uB^\flat)$ and $J=\curl B$ as $s=0$, where $s=dB^\flat-i_J\Omega$. Using $u\cdot B=C$ and $i_ui_J\Omega=-C'd\psi$, we see that
\begin{equation}
L_uB^\flat=i_us.
\end{equation}
For the converse, note that $|u|^2s=i_u(u^\flat\wedge s)+u^\flat\wedge i_us$ and
\begin{equation}
u^\flat\wedge s=-d(u^\flat\wedge B^\flat)+i_v\Omega\wedge B^\flat-(i_Ju^\flat)\Omega=(\Delta\psi+B\cdot v-J\cdot u)\,\Omega.
\end{equation}
Thus, in light of (\ref{eq:preGS0}), we deduce that on solutions of (\ref{eq:preGS})
\begin{equation}
s=-|u|^{-2}u^\flat\wedge L_uB^\flat,
\end{equation}
which completes the proof.
\end{proof}

In conclusion, the system of equations (\ref{eq:qsGS})-(\ref{eq:extra3}) describes the conditions that a quasi-symmetry $u$ and the corresponding flux surfaces $\psi$ must satisfy in magnetohydrostatics.

\section{Variational principle for the quasi-symmetric GS equation}

Two questions arise: (i) does (\ref{eq:qsGS}) have solutions, and (ii) how do we incorporate the supplementary conditions (\ref{eq:extra0})-(\ref{eq:extra3})?

In this section, we address principally the first question. 

There are a number of relevant results on the existence theory for semilinear elliptic PDEs, for example, Theorem 15.12 in \cite{GT} 
and Theorem 9.12 of \cite{Am}.  But even for the axisymmetric GS equation, there are regimes with no solutions \cite{AM}, and regimes with more than one solution.  We have to do more work to reach definitive conclusions.


In the meantime, however, we address here the question whether (\ref{eq:qsGS}) has a variational principle, because it would be one useful route to prove existence of a minimiser or some other critical point by variational means, cf.~\cite{BB} and to understand the set of solutions.  To that end we resort to the Helmholtz conditions, as formulated in \cite{O}. 

For a variational problem of the form $D{\mathcal L}_\psi = 0$ (often written $\delta{\mathcal L}[\psi]=0$) with ${\mathcal L}[\psi] = \int_Q L(\psi,\nabla \psi)\ dV$ (where $L$ is called the {\em Lagrangian} and may involve more derivatives) on smooth functions $\psi: Q \to \R$, the {\em Euler-Lagrange operator} $E$ on smooth functions $\psi$ is defined by writing
\begin{equation}
D{\mathcal L}_\psi v = -(E[\psi],v) 
\end{equation}
for all $v:Q\to \R$ (often written $\delta\psi$) satisfying suitable boundary conditions, where $(f,g) = \int_Q fg\ dV$ is the standard inner product on $L^2(Q,\R)$.  So the Euler-Lagrange equations are $E[\psi]=0$.

The Helmholtz conditions say that a differential operator $E$ on functions $\psi: Q \to \R$ is the Euler-Lagrange operator for some variational problem iff $(f,DEg) = (DEf,g)$ for all functions $f,g$ for which both sides are defined.  We write this as $DE^* = DE$ where $DE^*$ is the adjoint operator defined wherever $(DE f,g) = (f, DE^* g)$ makes sense, and refer to such operators as self-adjoint, but ignoring the question of equality of domains that is part of the standard definition.

A catch with applying the Helmholtz conditions is that the variational property is not preserved between equivalent equations, not even, for example, $\lambda E[\psi]=0$ and $E[\psi]=0$ where $\lambda$ is some non-zero function on $Q$. In fact the Helmholtz criterion for the lefthand side of (\ref{eq:qsGS}) as it stands implies the highly restrictive case $u\times v=0$ since $\Delta$ is self-adjoint. However, the axisymmetric and helical cases suggest the use of the factor $\lambda=|u|^{-2}$, so let us consider
\begin{equation}
\label{eq:qsGSf}
E[\psi] = |u|^{-2}\Delta\psi - |u|^{-4}u\times v\cdot \nabla \psi + |u|^{-4}u\cdot v\,C(\psi) + |u|^{-2}CC'(\psi) + p'(\psi).
\end{equation}

\begin{thm}\label{prop:vp}
$E$ given by (\ref{eq:qsGSf}) is the Euler-Lagrange operator for some variational problem $\mathcal{L}$ iff $L_uu^\flat=0$. In this case
\begin{equation}
\label{eq:LqsGS}
\mathcal{L}[\psi] = \int \left(\frac{1}{2|u|^2}(|\nabla \psi|^2-C(\psi)^2) + C(\psi)Y\cdot\nabla\psi - p(\psi) \right) dV,
\end{equation}
where $Y$ is a vector field such that $\div Y=u\cdot v/|u|^4$.
\end{thm}

\begin{proof}
Note first that the last three terms of $E$ are functions of $\psi$ only. Therefore their derivative is a multiplication operator, which is always self-adjoint. Thus the problem is reduced to just $E[\psi]=|u|^{-2}\Delta \psi-F\cdot\nabla\psi$, where $F=u\times v/|u|^4$. Now $E$ is a linear differential polynomial, and so $DE=E$, i.e.
\begin{equation}
DE=|u|^{-2}\Delta-F\cdot\nabla .
\end{equation}
Integration by parts shows that $DE^*g = \Delta(g|u|^{-2})+\div(gF)$ for the formal adjoint of $DE$. Using the identities $\Delta(g|u|^{-2})=\div\nabla(g|u|^{-2})=|u|^{-2}\Delta g+2\nabla|u|^{-2}\cdot\nabla g+g\Delta|u|^{-2}$ and $\div(gF)=g\,\div F+F\cdot\nabla g$, we arrive at 
\begin{equation}
DE^*=|u|^{-2}\Delta-(2G+F)\cdot\nabla - \div G,
\label{eq:DE*}
\end{equation}
where $G=-\nabla|u|^{-2}-F$. In other words, $G^\flat=(d|u|^2+i_udu^\flat)/|u|^4=L_uu^\flat/|u|^4$. Hence $DE^*=DE$ iff $G=0$, i.e.~$L_uu^\flat=0$.

The first term of (\ref{eq:LqsGS}) introduces $-|u|^{-2}\Delta\psi-\nabla|u|^{-2}\cdot\nabla\psi$ into the Euler-Lagrange operator. For $L_uu^\flat=0$, $\nabla|u|^{-2}$ reduces to $-F$ and so we recover the first two terms of (\ref{eq:qsGSf}). The third term of the Lagrangian yields $-C \div Y = -Cu\cdot v/|u|^4$, and the remaining terms easily restore the rest of (\ref{eq:qsGSf}).
\end{proof}

Note that existence of a vector field $Y$ such that $\div Y=u\cdot v/|u|^4$ may look difficult to satisfy, but it can be expressed equivalently as saying that $a=u^\flat\wedge du^\flat/|u|^4$ is exact, since $du^\flat=i_v\Omega$ and so $u^\flat\wedge du^\flat=(i_vu^\flat)\Omega$ and $a=di_Y\Omega$ accordingly. In this way, we see that it is not much of a restriction, because $a$ is automatically closed, being a top-form.

In the case of axisymmetry where $u\cdot v=0$, the variational functional (\ref{eq:LqsGS}) for $Y=0$ recovers the Lagrangian \cite{BB} for the usual GS equation in cylindrical polar coordinates.

In the case of helical symmetry, $u\cdot v/|u|^4=2l/(r^2+l^2)^2$, then $Y=-l/(r(r^2+l^2))\hat{r}$, so we derive a Lagrangian for the helical GS equation.

Recall, however, Theorem \ref{thm:uKf}, for $u$ satisfying the supplementary conditions. In this case, unfortunately, $L_uu^\flat=0$ is the condition for $u$ to be a Killing field, and thus in Euclidean space holds only if $u$ generates an orbit of $SE(3)$.  Then we are back to the axisymmetric GS equation (rejecting translations and the helical case because they do not have bounded flux surfaces).  This is where the second question comes in.  Can the extra conditions be incorporated as constraints in a variational principle?  Perhaps, as in the previous section, one should view the problem as a simultaneous system of equations for $\psi$ and $u$.

An alternative approach to extending the GS equation to general quasi-symmetry is to circle-average the MHS equation in the form $i_BdB^\flat=dp$ and derive an equation for $\psi$ with respect to the circle-averaged Riemannian metric.  Strangely, the resulting GS equation always has a variational principle.  But as in our analysis here, there are extra conditions that must be satisfied and it is not clear that they can.  This will be written in a separate publication.

\section{Perspectives}

Is every quasi-symmetry a Killing field?  At least in the Euclidean case?  Or at least if one requires magnetohydrostatics?  A starting point is to analyse the condition of Theorem \ref{thm:MHSqscon}. Or might there be some ``Kovalevskaya" examples\footnote{Recall that Kovalevskaya found non-axisymmetric integrable cases for the dynamics of a top.}? It is not even clear whether these questions are global or local in nature. The isometry condition $L_ug = 0$ is certainly a local one, and this at least hints that the questions may be local. If this is indeed the case a prolongation analysis based on the Cartan-Kuranishi theorem may be sufficient to provide definitive answers. We will report on such an analysis in a future publication. 

The main point of stellarators is to achieve confined guiding-centre trajectories without significant toroidal current.  Is quasi-symmetry compatible with this goal?  The toroidal current enclosed by a flux-surface is just $\int_\gamma B^\flat$ round any poloidal loop on the flux-surface.  It is conventionally written as $2\pi I(\psi)$.  Not surprisingly, $\int_{pol}B^\flat=0$ iff $\iota_J=0$.  
We see no incompatibility between this and quasi-symmetry, but it depends on there being some non-axisymmetric quasi-symmetries.

Quasi-symmetry may be too strong an ideal to aim for.  Weaker conditions would suffice for single-particle confinement.  Omnigenity \cite{H} is one such concept which just requires flux surfaces and the average drift in flux function for a guiding centre to be zero to leading approximation.  This is automatic for circulating particles of ZGCM on irrational flux surfaces but requires a condition for all bouncing particles \cite{LC} and for circulating particles on rational surfaces\footnote{This is not usually recognised, but follows by the same arguments as for bouncing particles.}.  Quasi-symmetry implies omnigenity, but perhaps not vice versa, so omnigenity would allow a bit more scope \cite{LC}.

More generally, one invariant torus for FGCM at each value of energy and magnetic moment will confine all those inside it.  This might be too weak an approach, however, because particle interactions would lead to exchange of energy and magnetic moment and drive them across the confining tori.

Alternatively, approximate quasi-symmetry may be enough.  This can be achieved to some order by near-axis expansions \cite{LSP} and we intend to address it in more detail.

Quasi-symmetry may also be too weak an ideal to aim for. Even if quasisymmetric field configurations do exist they may do a poor job of confining the hot alpha particles generated by thermonuclear burn. The issue is such particles have much larger gyroradii than bulk plasma particles. In the best case, confinement properties of alpha particles might be well-captured by guiding-centre theory with higher-order corrections, in which case it would be interesting to study possible hidden symmetries of these high-order terms. In the worst case guiding-centre theory is useless for describing  alpha particle orbits, and other approaches, perhaps based on a more brute-force optimization approach, should be pursued.

Finally, we have restricted attention to symmetries of FGCM of the form $U=(u,0)$ with $u$ a vector field on guiding-centre position, not involving the parallel velocity.  Might there be parallel-velocity-dependent symmetries that render FGCM integrable? Likewise we have concentrated here entirely on the context of Hamiltonian symmetries.  Might there be relevant non-Hamiltonian symmetries? This is work in progress.

\section*{Acknowledgements}

This work was supported mainly by a grant from the Simons Foundation/SFARI (601970, RSM), and partly by the National Science Foundation under Grant No.~DMS-1440140 while the first and third authors were in residence at the Mathematical Sciences Research Institute in Berkeley, California, during the Fall 2018 semester. Research presented in this article was supported by the Los Alamos National Laboratory LDRD program under project number 20180756PRD4.

We are grateful to other members of the MSRI semester for useful comments, and of the Simons collaboration team, in particular to Adam Golab for researching the literature on existence results for solutions of the GS equation, and discussion on the supplementary conditions for its quasi-symmetric generalisation. We also want to thank Miles Wheeler for helpful suggestions about existence of solutions of the quasi-symmetric GS equation.

\section*{Appendix: Additions of electrostatic potential and relativity}

To add the effect of an electrostatic potential $\Phi$ to the theory of this paper, add $e\Phi(q)$ to $H$ in both equations (\ref{eq:H2}) and (\ref{eq:H7}).  The drift equations (\ref{eq:drift})-(\ref{eq:vdrift}) gain additional terms $b\times \nabla \Phi /\widetilde{B}_\pl$ and $-e\nabla \Phi / \widetilde{B}_\pl$ respectively.
The conditions of Theorems~\ref{thm:main} and \ref{thm:alt} need augmenting by $L_u\Phi=0$.

To add relativistic effects, one simply replaces $p=mv$ by $p=\gamma m v$ with Lorentz factor $\gamma = (1-|v|^2/c^2)^{-1/2}$, and the kinetic energy in $H$ by $c\,(m^2c^2 + |p|^2)^{1/2}$.  This gives the particle motion with respect to proper time.  Likewise the magnetic moment changes to $\mu = p_\perp^2/(2m|B|)$ and the kinetic part of the guiding-centre Hamiltonian to $c\,(m^2c^2+p_\pl^2 + 2m\mu|B|)^{1/2}$.  The conditions for quasi-symmetry are unchanged.

Alternatively, taking a fully space-time view and allowing general time-dependent electric and magnetic fields, the equation of motion is
\begin{equation}
\frac{dp}{d\tau} = -\,e\,i_U\!F,
\end{equation}
where $F$ is the Faraday 2-form, $U$ is the contravariant 4-velocity, $p=mU^\flat$ is the covariant 4-momentum, and $\tau$ is proper time for the particle. In a time-space coordinate system $(t,x,y,z)$ with locally Minkowski metric $ds^2 = -c^2 dt^2 + dx^2 + dy^2 + dz^2$,
\begin{equation}
F = B_x dy\wedge dz  + B_y dz \wedge dx + B_z dx \wedge dy + E_x dx \wedge dt + E_y dy \wedge dt + E_z dz \wedge dt
\end{equation}
and $U=\gamma(c,v^x,v^y,v^z)$. Noting that two of Maxwell's equations are equivalent to saying $F$ is closed, the motion of the charged particle can be written in Hamiltonian form with
\begin{align}
H &= \frac{|p|^2}{2m} \\
\omega &= -d\vartheta - eF,
\end{align}
where $\vartheta$ is the tautological 1-form on $T^*Q$ for space-time $Q$.  The number of degrees of freedom is now 4.  Guiding-centre reduction can still be performed resulting in a 3DoF system.  Integrability would now require two further integrals beyond the Hamiltonian.  It would be interesting to find out whether time-translation symmetry can be replaced.

\end{document}